\documentclass[onecolumn]{IEEEtran}

\usepackage{amsthm}
\usepackage{amsmath,amsfonts,amssymb,rotating,mathrsfs,silence,bbm,times,float}
\usepackage{mathtools}
\usepackage[mathscr]{eucal}
\usepackage{enumerate}
\usepackage{hyperref} 
\usepackage{graphicx} 
\usepackage{color}
\usepackage{enumitem}
\usepackage{epsfig}
\usepackage{graphicx}
\usepackage[table]{xcolor}
\usepackage{cite}
\usepackage{soul}

\usepackage[normalem]{ulem}
\newcommand\nc\newcommand
\nc\bfa{{\boldsymbol a}}\nc\bfA{{\boldsymbol A}}\nc\cA{{\mathscr A}}
\nc\bfb{{\boldsymbol b}}\nc\bfB{{\boldsymbol B}}\nc\cB{{\mathscr B}}
\nc\bfc{{\boldsymbol c}}\nc\bfC{{\boldsymbol C}}\nc\cC{{\mathscr C}}
\nc\bfd{{\boldsymbol d}}\nc\bfD{{\boldsymbol D}}\nc\cD{{\mathscr D}}
\nc\bfe{{\boldsymbol e}}\nc\bfE{{\boldsymbol E}}\nc\cE{{\mathscr E}}
\nc\bff{{\boldsymbol f}}\nc\bfF{{\boldsymbol F}}\nc\cF{{\mathscr F}}
\nc\bfg{{\boldsymbol g}}\nc\bfG{{\boldsymbol G}}\nc\cG{{\mathscr G}}
\nc\bfh{{\boldsymbol h}}\nc\bfH{{\boldsymbol H}}\nc\cH{{\mathscr H}}
\nc\bfi{{\boldsymbol i}}\nc\bfI{{\boldsymbol I}}\nc\cI{{\mathcal I}}
\nc\bfj{{\boldsymbol j}}\nc\bfJ{{\boldsymbol J}}\nc\cJ{{\mathscr J}}
\nc\bfk{{\boldsymbol k}}\nc\bfK{{\boldsymbol K}}\nc\cK{{\mathscr K}}
\nc\bfl{{\boldsymbol l}}\nc\bfL{{\boldsymbol L}}\nc\cL{{\mathscr L}}
\nc\bfm{{\boldsymbol m}}\nc\bfM{{\boldsymbol M}}\nc{\cM}{{\mathscr M}}
\nc\bfn{{\boldsymbol n}}\nc\bfN{{\boldsymbol N}}\nc\cN{{\mathscr N}}
\nc\bfo{{\boldsymbol o}}\nc\bfO{{\boldsymbol O}}\nc\cO{{\mathscr O}}
\nc\bfp{{\boldsymbol p}}\nc\bfP{{\boldsymbol P}}\nc\cP{{\mathscr P}}\nc\eP{{\EuScriptP}}\nc\fP{{\mathfrak P}}
\nc\bfq{{\boldsymbol q}}\nc\bfQ{{\boldsymbol Q}}\nc\cQ{{\mathscr Q}}
\nc\bfr{{\boldsymbol r}}\nc\bfR{{\boldsymbol R}}\nc\cR{{\mathscr R}}
\nc\bfs{{\boldsymbol s}}\nc\bfS{{\boldsymbol S}}\nc\cS{{\mathscr S}}
\nc\bft{{\boldsymbol t}}\nc\bfT{{\boldsymbol T}}\nc\cT{{\mathscr T}}
\nc\bfu{{\boldsymbol u}}\nc\bfU{{\boldsymbol U}}\nc\cU{{\mathscr U}}
\nc\bfv{{\boldsymbol v}}\nc\bfV{{\boldsymbol V}}\nc\cV{{\mathscr V}}
\nc\bfw{{\boldsymbol w}}\nc\bfW{{\boldsymbol W}}\nc\cW{{\mathscr W}}
\nc\bfx{{\boldsymbol x}}\nc\bfX{{\boldsymbol X}}\nc\cX{{\mathscr X}}
\nc\bfy{{\boldsymbol y}}\nc\bfY{{\boldsymbol Y}}\nc\cY{{\mathscr Y}}
\nc\bfz{{\boldsymbol z}}\nc\bfZ{{\boldsymbol Z}}\nc\cZ{{\mathscr Z}}

\usepackage[T1]{fontenc}
\usepackage{cmbright}
\newcommand{\tr}[1]{\mathrm{tr}_{\mathbb{K}/F_{#1}}}
\newcommand{\zz}[1]{{\color{black}#1}}

\newtheorem{theorem}{Theorem}[section]
\newtheorem{lemma}[theorem]{Lemma}

\newtheorem{proposition}[theorem]{Proposition}

\newtheorem{definition}[theorem]{Definition}

\newtheorem{construction}{Construction}[section]

\theoremstyle{remark}

\allowdisplaybreaks

\DeclareMathOperator{\rk}{rank}

\begin{document}
	\title{Explicit constructions of MSR codes for clustered distributed storage: The rack-aware storage model}
	
\author{\IEEEauthorblockN{Zitan Chen} \hspace*{1in}
\and \IEEEauthorblockN{Alexander Barg}}
\maketitle	

{\renewcommand{\thefootnote}{}\footnotetext{

\vspace{-.2in}
 
\noindent\rule{1.5in}{.4pt}

{The results of this paper were presented in part at the Allerton Conference on Communication, Control and Computing, Monticello, IL, Oct. 2018. 

Z. Chen is with Department of ECE and ISR, University of Maryland, College Park, MD 20742. Email: chenztan@gmail.com
His research was supported by NSF grant CCF1618603.

A. Barg is with Dept. of ECE and ISR, University of Maryland, College Park, MD 20742 and also with IITP, Russian Academy of Sciences, 127051 Moscow, Russia. Email: abarg@umd.edu. His research was supported by NSF grants CCF1618603 and CCF1814487.}}
}
\renewcommand{\thefootnote}{\arabic{footnote}}
\setcounter{footnote}{0}

\begin{abstract} The paper is devoted to the problem of erasure coding in distributed storage. We consider a model of storage that assumes that nodes are organized into equally sized groups, called racks, that within each group the nodes can communicate freely
without taxing the system bandwidth, and that the only information transmission that counts is the one between the racks. This assumption implies that the nodes within each of the racks can collaborate before providing information to the failed node. The main emphasis of the paper is on code construction for this storage model. We present an explicit family of MDS array codes that support recovery of a single failed node from any number of helper racks using the minimum possible amount of inter-rack communication
(such codes are said to provide optimal repair). The codes are constructed over finite fields of size
comparable to the code length.

We also derive a bound on the number of symbols accessed at helper nodes for the purposes of repair, and construct a code family that
approaches this bound, while still maintaining the optimal repair property.

Finally, we present a construction of scalar Reed-Solomon codes that support optimal repair for the rack-oriented storage model.

\end{abstract}

	\section{Introduction}
Erasure codes increase reliability and efficiency of distributed storage by supporting the recovery of the data on failed nodes
under the restriction of low repair bandwidth, i.e., limited amount of information downloaded from other nodes for the purposes of the repair. 
This problem was initially introduced in the well-known paper \cite{dimakis2010network} which cast the capacity problem of distributed storage as
a network coding problem where the necessary conditions for the repair of failed nodes were derived 
by considering the information flow in the network that occurred in the course of repair. These conditions imply a bound on the minimum number
of symbols required for repair of a single failed node, which is known as the cut-set bound on the repair bandwidth. Paper 
\cite{dimakis2010network} further considered a variety of data coding schemes that optimize either storage or repair bandwidth, as well
as the tradeoff between these two quantities. In this paper we limit ourselves to minimum-storage regenerating (MSR) codes, or, equivalently, to
MDS codes with optimal repair. We further restrict our attention to the task of exact repair as opposed to mode general functional repair 
\cite{dimakis2010network}.

Initially the repair problem was formulated for the so-called centralized repair model which assumes that the failed nodes are repaired by a single data collector that
receives information from the helper nodes and performs the recovery within a single location, having full access
to all the downloaded information and the intermediate results of the calculations \cite{Cadambe13,Ye16,rawat2016centralized}. Another well-known
model assumes cooperative repair, when the failed nodes are restored at different physical locations, and the information downloaded to each of them as well as the exchange of intermediate results between them are counted toward the overall repair bandwidth \cite{Kermarrec11,Shum13,Li14,Ye19a}.

The problems of centralized and cooperative repair have been addressed in multiple recent papers, and there are
explicit constructions of optimal-repair regenerating codes that cover the entire range of admissible parameters, require small-size
ground alphabet compared to the length $n$ of the encoding block, and attain the smallest possible repair bandwidth
\cite{Rashmi11,Tamo13},\cite{Ye16},\cite{ye2017explicit,Sasid16,clay18},\cite{LiTangTian18} (more references are given in a recent survey \cite{Balaji18}).
The availability of optimal constructions has motivated a shift of attention toward studying data recovery not only under communication,
but also {\em connectivity constraints}, in other words, storage models in which communication cost between nodes differs depending
on their location in the storage cluster. One of the simple extensions from the basic setting of homogeneous storage 
suggests that the nodes are joined into several groups
(clusters), and repair of a node can be based on information from both the nodes within its own group and from nodes in the other groups.
This permits to differentiate between communication within the cluster and the inter-cluster downloads, and the natural assumption
is that the former is easier (contributes less to the repair bandwidth) than the latter.

Erasure coding for clustered architectures was introduced several years ago and affords several variations. One of the first questions analyzed 
for heterogeneous storage models was related to repair under the condition that the system contains a group of nodes, downloading information from which contributes more to the repair bandwidth than downloading the same 
amount of information from the other nodes \cite{Akhlagi10}. Later works \cite{Gaston2013,Pernas2013} observed that a more realistic version of 
non-homogeneous storage should assume that the cost of downloading information depends on the relative location of the failed node in the system. 
In this case, downloading information from the group that contains the failed node (also called the {\em host group}) contributes less to the 
cost than inter-cluster downloads. The authors of \cite{Gaston2013,Pernas2013} have assumed that the storage is formed of two 
clusters and derived versions of the cut-set bound for the minimum repair bandwidth. The two-cluster model was further developed in recent papers 
\cite{Sohn18a,Sohn18} which assumed that the encoded data is placed in a number of clusters (generally more than two), and derived a cut-set type 
bound on the repair bandwidth for this case. Moreover, \cite{Sohn18} showed existence of optimal-repair codes for their model, and 
\cite{Sohn18a} gave an explicit construction of MDS array codes of rate 1/2 for storage with two clusters. We also mention 
\cite{Sahraei2017a,Ye16,PAM18,Tebbi2014} which discuss other variations of clustered storage architectures and are less related to our work.

\subsection{Rack-aware storage} The model that we address in this work assumes that $k$ data blocks are encoded into a codeword of length $n={\bar n} u$ and stored across $n$ nodes. The nodes are organized into $\bar n$ groups, also called racks.  Suppose that a node has failed and call the
rack that contains it the {\em host rack}. To perform the repair, the system downloads information from the nodes in the host rack (called below {\em local nodes}), as well as information from the other racks.
The rack-oriented storage model is distinguished from the other clustered storage architectures in that the 
information from nodes that share the same rack, can be processed before communicating it to the failed node. Communication within the racks, including the host rack, does not incur any cost toward the repair bandwidth. The main benefit of rack-aware coding is related to
reducing the bandwidth required for repair compared to coding for homogeneous storage.

This model was introduced in \cite{hu2016double,hu2017optimal}. 
Specifically, the authors of \cite{hu2016double} derived a version of the cut-set bound
of \cite{dimakis2010network} adapted for this case and showed existence of minimum-storage codes with optimal repair for the rack model. A more expanded study of codes for this model, both for the minimum-storage and minimum-bandwidth 
scenarios, was undertaken in a recent paper \cite{hou2018rack}, which showed existence of codes with optimal repair
bandwidth for a wide range of parameters. At the same time, there are very few explicit constructions of MSR codes for this model 
known in the literature. We mention \cite{hu2017optimal} which presented such codes for 3 racks and for the case when the number
of parity symbols of the code $r:=n-k=\bar n.$ 

\subsection{Main results}
In this paper we present constructions of minimum-storage regenerating codes for the rack-aware storage model that have optimal repair
bandwidth and cover all admissible parameters, such as the code rate $k/n$, the size and number of the racks. 
The only restriction that we assume is the natural condition that the racks are of equal size $u$ and that the codeword is written on $\bar n$
racks such that $u$ symbols of the field are placed on each of them.
This assumption is also consistent with the literature \cite{hu2016double,hou2018rack}.

We present two families of MDS array codes that support optimal repair in the rack model. The first family gives 
an explicit construction of optimal-bandwidth codes for repairing a {\em single node} from the nodes located in $\bar d$ helper racks for any $
\lfloor k/u\rfloor\le\bar d\le\bar n-1$. The underlying
finite field of our construction is of size at most $n^2/u$ where $u$ is the size of the rack, and the node size (sub-packetization) equals $l\approx 
(\bar d-\frac ku)^{n/u}.$ The construction is phrased in terms of the parity-check equations of the code, as in \cite{Ye16,ye2017explicit}, and
relies on the multiplicative structure of the field to account for the rack model considered here.

The second code family constructed in this paper, in addition to optimal repair, addresses the question of reducing the number of symbols accessed
on each of the helper racks. The code construction is presented in two steps. First, we present a new family of optimal-access codes for the standard repair problem (homogeneous storage), constructing codes
 with arbitrary repair degree $d, k\le d\le n-1$ over a field $F$ of size at least $d-k+1.$ 
These parameters
are similar to optimal-access codes constructed in \cite{Ye16}, and in fact require a slightly larger field $F$. At the same time, 
the new construction can be modified for the rack model, resulting in codes with low access. 

We also present a family of (scalar) Reed-Solomon codes that can be optimally repaired in the rack model. The construction is a modified
version of the RS code family constructed in \cite{Tamo18} for the case of homogeneous storage.

Apart from the code constructions, we examine the structure of codes with optimal repair or optimal access for the rack model. Because of intra-rack processing, the definition 
of optimal access is not as explicit as in the homogeneous case. We prove a lower bound on the number of accessed symbols for codes that
support optimal repair. At the same time, the codes that we construct fall short of attaining this bound, and it is not clear what is the correct 
value of this quantity.

Finally, we derive a lower bound on the node size for optimal-repair codes in the rack model, modifying for this purpose the approach of the recent work \cite{balaji2017tight}, where a similar bound was proved for the homogeneous case. 

\section{Problem statement and structural lemmas}
Assume that the data file of size $M$ is divided into $k$ blocks and encoded
using an array code $\cC$ of length $n$ over some finite field $F$. Each symbol of the codeword is represented by an $l$-dimensional 
vector over $F$ and is placed on a separate storage node. 
We assume that the code is MDS, i.e., the entire codeword
can be recovered from any $k$ of its coordinates (from the encoding stored on any $k$ out of the $n$ nodes).
According to the cut-set bound of \cite{dimakis2010network}, the amount of information required for repair of
a single node from $d$ helper nodes satisfies the inequality
   \begin{equation}\label{eq:cutset}
     \beta(n,k)\ge \frac{d l}{d-k+1}, 
        \end{equation}
where $k\le d\le n-1.$ 

Suppose that information is encoded with an MDS array code $\cC$ of length $n=\bar n u$ over a finite field $F.$ If the size of the code is $q^{kl},$ we refer to it as a $\cC(n,k,l)$ code. The set of 
nodes $[n]=\{1,2,\dots,n\}$ is partitioned into $\bar n$ subsets (racks) of size $u$ each. Accordingly, the coordinates of the codeword $c\in\cC$ 
are partitioned into segments of length $u$, and we label them as 
$c_t,t=1,\dots,n$, where
$t=(m-1)u+j, 1\le m\le \bar n,\; 1\le j\le u$. We do not distinguish between the nodes and the coordinates of 
the codeword, and refer to both of them as nodes. Each node is an element in $F^l,$ and when needed, we denote its entries as 
$c_{t,j},j=1,\dots,l.$

Denote by $\cR\subset\{1,\dots,\bar n\}$ the set of $\bar d$ helper racks and let $m^\ast$ be the index of the host rack.
To repair the failed node, information is generated in the helper racks and is combined with the contents of the local nodes to perform the repair. This is modeled by computing a linear function of the contents of the nodes within each helper rack (the function depends on the contents of all the nodes in the rack, and can in principle also depend on the rack index), and sending this information to rack $m^\ast.$

\begin{definition}[{\sc Repair scheme}]\label{def:repair} Let $\cC(n,k,l)$ be an array code. 
Suppose that node $c_{m^\ast u+j^\ast}$ is erased (has failed). To recover the lost data, we rely on the values of the symbols in coordinates
$c_{iu+j},$ where $i\in \cR$  and $j=1,\dots,u.$
A {\em repair scheme} $\cS$ with repair degree $\bar d\le \bar n-1$ is formed of $\bar d$ functions $f_{i}:F^{ul}\to F^{\beta_{i}},t=1,\dots,\bar d$ and a function 
$g:F^{\sum_{i\in \cR}^{\bar d} \beta_{i}}\times F^{(u-1)l}\to F^l.$ For a given $i\in \cR$ the function $f_{i}$ maps $c^{(i)}$ (the nodes in rack $i$) to some $\beta_{i}$ symbols of $F.$ The function $g$ accepts these symbols together
with the available nodes in the {host rack} as arguments, and returns the value of the failed node:
   $$
   g(\{f_{i}(c_{i u+j},1\le j\le u),i\in\cR\},\{c_{m^\ast u+j}, j\in\{1,\dots,u\}\backslash \{j^\ast\}\})=c_{m^\ast u+j^\ast}.
   $$ 
In general the function $f_{i},i\in\cR$ depends on $i,m^\ast$ and $j^\ast$, and the function $g$ depends on $\cR,m^\ast,j^\ast.$

The quantity $\beta(\cR,m^\ast,j^\ast)=\sum_{i\in\cR}\beta_{i}$
is called the repair bandwidth of the node $c_{m^\ast u+j^\ast}$ from the helper racks in $\cR$ and from the available nodes in the host rack $m^\ast.$
\end{definition}

The repair scheme can be defined in a more general way: for instance, each of the functions $f_{i_t}$ that form the information
downloaded by the failed node could depend on the entire set $\cR,$ and the dependence of the function $g$ on $\cR$ could be not just
through the downloaded information. At the same time, all our results as well as all the results in the earlier literature are well described by this definition, which therefore suffices for our purposes. If the functions $f_{i_t},g$ are $F$-linear, the repair scheme itself is called {\em linear}. Only such schemes will be considered below.

Let
   $$
   \beta(n,k,u):=\min_{\cC\subset F^{nl}}\max_{\cR,m^\ast,j^\ast} \beta(\cR,m^\ast,j^\ast)
   $$
where the minimum is taken over all $(n,M=q^{kl})$ MDS array codes and the maximum over the index of the host rack, the failed node in the rack, and the choice of the set of the helper racks $\cR.$  To rule out the trivial case, we assume throughout that $k\ge u.$

\subsection{Optimal repair}
Suppose that $k=\bar k u+v,$ where $0\le v\le {u}-1.$	
A necessary condition for successful repair of a single node is given by a version of the cut-set bound \cite{hu2016double}, \cite{hou2018rack} which states that
for any $(n,k,l)$ MDS array code, the (inter-rack) repair bandwidth is at least 
	\begin{align}
		\beta(n,k,u)\ge \frac{{\bar d}l}{\bar d-\bar k+1} 
		 \label{eq:rack-bound}
	\end{align}
The code that attains this bound with equality is said to have the optimal repair property.	
	
The arguments below are based on the following obvious (and well-known) observation.
\begin{lemma}\label{lemma:MDS} Let $\cC(n,k,l)$ be an MDS array code. Suppose that a failed node is repaired using a set $\cI,|\cI|=d$ of helper nodes. The number of symbols  of $F$ downloaded for the repair task from any subset $\cI'\subset \cI$ of size $|\cI'|=d-k+1$ is at least $l.$ 
\end{lemma}
To prove this it suffices to observe that, because of the MDS property, no subset of $k-1$ nodes carries any information about the value of any 
other node. 

We note that this lemma applies to the rack model (i.e., allowing processing of the information obtained from the nodes in $\cI$). It also applies
if the count of downloaded symbols is replaced by the count of symbols {\em accessed} on the helper nodes.

The next statement, called the {\em uniform download property}, is well known for the case of homogeneous
storage. Its proof for the rack-aware storage is not much different, and is given for completeness in Appendix \ref{sec:uniform}.
\begin{proposition}\label{prop:uniform}
Let $\cC$ be an MSR code and suppose that $\bar k>1.$ Let $\cR$ be the set of helper racks used to repair a single failed node.
Then $\beta_{i}={l}/{(\bar d-\bar k+1)},i\in \cR.$
\end{proposition}
 We note that both the bound \eqref{eq:rack-bound} and this proposition can be generalized
to the case of $2\le h\le r$ failed nodes located on the same rack without any difficulty; for instance, the bound takes the form $\beta\ge 
\frac{h\bar d l}{\bar d-\bar k+1}.$

Next, observe that if $k$ is divisible by the rack size $u$, then any MSR code for the
standard model will be optimal for the rack model, i.e., cooperation between the nodes within the rack does
not help to reduce the repair bandwidth (this has been first observed in \cite[Thm. 4]{hou2018rack}). 
\begin{proposition} Let $k=\bar ku,$ and let $\cC$ be an MSR code of length $n=\bar n u$ with optimal repair of a single node 
for the homogeneous storage model. Then $\cC$ attains the cut-set bound \eqref{eq:rack-bound} for repair of any single node 
in the rack-aware model.
\end{proposition}

{\em Proof:}
Take
an MSR code of length $n$ and assume that $v=0.$ 
Suppose that the number of helper nodes is $d$, and this includes the $u-1$ local nodes. By \eqref{eq:cutset}, the repair bandwidth
necessary equals
    $
    \frac{d }{d-k+1} l.
    $
In accordance with the model, take $d=\bar d u+(u-1)$, then
  \begin{equation}\label{eq:du}
  \frac{d }{d-k+1}l=\Big(\frac{\bar d}{\bar d-\bar k+1}+\frac{u-1}{d-k+1}\Big)l
  \end{equation}
and this achieves the bound \eqref{eq:rack-bound} if the second term is discounted (which is possible because of the uniform download property and because intra-rack communication is free).
\qed

Note that in the case of $v\ne 0,$ optimal codes for the rack model perform repair using a strictly smaller repair bandwidth than 
optimal codes for the homogeneous model. This also suggests that the number of symbols downloaded from a helper rack is strictly smaller than
the number of accessed symbols, i.e., intra-rack processing is necessary for optimal repair (this will be made rigorous once we establish Prop.~\ref{prop:optimal-access} below).

For reader's convenience, let us summarize the code parameters: We consider $(n,k,l)$ array codes used in a system where the nodes are arranged in racks of size $u$. The codes are designed to repair a single node. We further assume that $n=\bar n u, k=\bar k u+v,$ where $0< v\le u-1$, and the number of helper racks is $\bar d,$ where $\bar k\le\bar d\le\bar n-1$. We also use the notation $r=n-k,\bar r=\bar n-\bar k$ for the number of parity nodes and parity racks, respectively. Finally, to shorten the formulas we denote
  \begin{gather*}
  s=d-k+1, \quad\bar s=\bar d-\bar k+1,
  \end{gather*}
where $d$ is the total number of helper nodes accessed for repair, and $\bar d$ is the {\em repair degree}, i.e., number of helper racks (not counting the host rack).

\subsection{Optimal access} 
Some of the constructions of codes for the homogeneous case have the additional property that the information accessed on the helper nodes is the same as the information
that is downloaded by the helper node (no processing is performed before downloading). This property, also called {\em repair by transfer},
reduces the implementation overhead, and is therefore desirable in the code construction. Structure and constructions of optimal access (OA) codes 
for the homogeneous case were addressed in \cite{Tamo14,clay18,ye2017explicit} among others.

\begin{definition} Let $\cC(n=\bar n u,k,l)$ be a code that supports optimal repair of a single failed node with repair degree $\bar d.$ Suppose that each of the helper racks provides $l/\bar s$ field symbols and these symbols are generated by accessing the smallest possible number
of symbols of the nodes in the rack. In this case we say that $\cC$ has the OA property.
\end{definition}

To motivate this definition, we draw an analogy with the homogeneous case. In this case, on account of the bound \eqref{eq:cutset} and the uniform download property, the system accesses $l/s$ symbols at each of the helper nodes, and these symbols are downloaded to accomplish the repair.
As a consequence, a group of $u>1$ helper nodes provides $ul/s$ symbols.
This observation also extends to the rack-aware model in the case that $u|k.$ Indeed, in this case the number of symbols downloaded from, and accessed on, each rack equals $l/\bar s=ul/s.$ 

In the next proposition (proved in Appendix~\ref{app:optimal-access}) we derive a lower bound on the number of accessed symbols and establish the {\em uniform access condition}.
\begin{proposition}\label{prop:optimal-access} Let $\cC$ be an $(n,k,l)$ optimal-repair MDS array code for the rack model with repair degree $\bar d\ge \bar k+1$ \zz{and $u\leq k$}.
The number of symbols accessed on the helper racks for repair of a single node satisfies
   \begin{equation}\label{eq:bound-access}
   \alpha\ge \frac{\bar d u l}{s}.
   \end{equation}
Equality holds if and only if the number of symbols accessed on node $e$ satisfies $\alpha_{m,e}=l/s$ for all $m\in\cR;\,e=1,\dots,u.$
\end{proposition}
As noted above, if $u|k,$ the symbols accessed on the helper nodes can be downloaded without processing, accounting for optimal repair. At the
same time, if $u\nmid k,$ and the code meets the bound \eqref{eq:bound-access}, then processing is necessary because $\bar d u l/s$
is strictly greater than the optimal bandwidth in \eqref{eq:rack-bound}.

	\subsection{A lower bound on the sub-packetization of rack-aware optimal-access MSR codes}

In this section we present a lower bound on the value of the node size in MSR codes for the rack model, which will be implicitly assumed throughout
without further mention. Similarly to \cite{Tamo14,balaji2017tight}, we limit ourselves to 
systematic codes and linear repair schemes. Let $\cC$ be an $(n=\bar{n}u,k=\bar{k}u,l)$ systematic optimal-access MSR array code over $F$.
Let $A=(A_{ij})$ be the $((n-k)l\times kl)$ encoding matrix of $\cC$; in other words, the parity symbols $c_{k+i},i=1\ldots,r=n-k$ are obtained 
from the data symbols $c_j, j=1,\dots,k$ according to the relation
	\begin{align}
		c_{k+i} = \sum_{j=1}^{k}A_{i,j}c_{j}, \label{eq:parity}
	\end{align}
where each $A_{i,j}$ is an $l \times l$ invertible matrix over $F$. Assume without loss of generality that the $k$ systematic nodes are
located on racks $1,\dots,\bar k$, called systematic racks below. Racks $\bar k+1,\dots, \bar n$ will be called parity racks. 
Let $\bfc_m = (c_{(m-1)u+1},\ldots, c_{mu})^T$ be the data vector stored in the $m$-th rack, $1\le m\le \bar{k}$, where each component is an $l$-vector over $F.$ Suppose for definiteness that the failed node is located in rack $m_1$, where $1\le m_1\le \bar{k}$. Suppose further that the set of 
$\bar{d}$ helper racks is formed of the remaining $\bar{k}-1$ systematic racks and some $\bar{s}=\bar{d}-\bar{k}+1$ parity racks.
	
We assume throughout that the repair scheme is independent of the index of the failed node in its rack.

The main result of this section is given in the following theorem, whose proof is modeled on the result of \cite{balaji2017tight}.
	\begin{theorem} \label{thm:bound} Let $\cC$ be an $(n=\bar{n}u,k=\bar{k}u,l)$ {optimal-access} MSR array code, $k\ge u$, and let $\bar d,\bar{k}\le\bar{d}\le\bar{n}-1$ be the size of the helper set $\cH$. Suppose further that there is a linear
repair scheme that supports repair of a single failed node from any $\bar d$ helper racks. 
	
	(a)  Suppose that the repair scheme depends on the choice of the helper racks as well as on the index of \zz{the host rack}.
	Then
		\begin{align}
			l \geq \min\{\bar{s}^{(\bar{n}-1)/{s}},\bar{s}^{\bar{k}-1}\}, \label{eq:bound1}
		\end{align} 
		where $\bar{s}=\bar{d}-\bar{k}+1$ and $s=\bar{s}u$. 

	(b) Suppose that the repair scheme depends on the index of the host rack but not on  the choice of the helper racks, then
		\begin{align}
			l \geq \min\{\bar{s}^{{\bar{n}}/{s}},\bar{s}^{\bar{k}-1}\}. \label{eq:bound2}
		\end{align} 
	\end{theorem}
A proof of this theorem is given in the Appendix. Here let us make the following remark. The theorem is proved under the assumption that $u|k,$ in which case any optimal-access MSR code for the homogeneous storage model supports optimal repair for the rack model. The smallest possible value of sub-packetization for such codes is $\zz{l=r^{\lceil\frac {n-1} r \rceil}}$ \cite{balaji2017tight,ye2017explicit}. Thus, this theorem says that it is possible that there
exist optimal-access rack codes that have smaller node size than OA codes for homogeneous storage {\em even in the case when $k$ is a multiple of $u$.}

%

\section{Rack-aware codes with optimal repair for all parameters}
\label{sec:RackCodes}
Let $\bar{s}=\bar{d}-\bar{k}+1$ and let {$F, |F|> \bar s n$} be
a finite field. 
The code that we construct is formed as an $F$-linear array MDS code $\cC$ of length $n,$ dimension $k$, and sub-packetization $l=\bar s^{\bar n}.$
 We denote a codeword of $\cC$ by $(c_1,c_2,\dots,c_n),$
where $c_i=(c_{i,1},\dots,c_{i,l})$ for all $i=1,\dots,n$.  Suppose that {$\bar s n \,|\, (|F|-1)$} and let $\lambda\in F$ be an element of multiplicative order {$\bar s n$}.
Finally, given $j\in\{0,1,\dots l-1\},$ consider the base {${\bar s}$} expansion $j=(j_{\bar n},j_{\bar n-1},\dots,j_1)$ and let
   \begin{equation}\label{eq:jpa}
   j(p,a):=(j_{\bar n},\dots,j_{p+1},a,j_{p-1},\dots,j_1),
   \end{equation}
    where {$0\le a\le \bar s-1.$}

\vspace*{.05in}\begin{construction}\label{Construction1} 
{\rm Define an $(n,k,l=\bar s^{\bar n})$ code $\cC=\{\bfc=(c_{i,j})_{1\le i\le n;0\le j\le l-1}\}$ defined by the following set of $rl$ parity-check equations over $F$:
  \begin{gather}\label{eq:code}
  \sum_{e=1}^{\bar n}\lambda^{t((e-1)\bar {s} +j_e)}\sum_{i=1}^u\lambda^{t(i-1)\bar {s}\bar n}
    c_{(e-1)u+i,j}=0
  \end{gather}
for all $t=0,\dots,r-1; j=0,\dots,l-1.$
}
\end{construction}

\vspace*{.05in}
We will show that the code defined in \eqref{eq:code} is an MDS code that has the smallest possible repair bandwidth according to the bound \eqref{eq:rack-bound}. 
Before stating the main theorem that proves these claims let us comment on the origin as
well as the new elements in this construction. The code is formed of two levels, the algebraic one, which accounts for the
repair of a node in any {\em fixed} rack, and a stacking construction which makes the code universal (i.e., rack-independent).
The first of these two levels was originally introduced in the MSR code construction in \cite{Ye16} and then developed in \cite{ye2017explicit} to account for cooperation between the failed nodes. Here we expand on this idea by
permitting the nodes within each of the racks to cooperate before passing on the information to the repair center. 
This is made possible by exploiting the multiplicative structure of the field $F$, and represents a new idea
introduced in this work.

The second component, which accounts for the universality property, has been isolated and discussed in 
\cite[Sec.~V.C]{Ye19a} (where it was called the $\odot$ operation), and is by now standard. It is
based on the idea of representing the index $j$ of the node coordinate as an $\bar {s}$-ary number, and we do not spend much space on it here.

\begin{theorem} Let {\color{black}{$\bar k\le\bar{d}\le \bar n-1.$}} The $(n,k,l=\bar s^{\bar n})$ code $\cC$ defined by the parity-check equations \eqref{eq:code} is an MDS code that
supports optimal repair of any single node from any $\bar d$  helper racks,  under the rack-aware storage model. 
\end{theorem}
\begin{IEEEproof}
We begin with proving the part of the claim about the repair properties of the code $\cC.$ Suppose that the index of the
rack that contains the failed node is $p\in\{1,\dots,\bar n\}.$  We have $\bar ru=r+v$ and since $1\le v\le {u}-1,$ $\zz{(\bar r-1)u<r-1}.$ 
Rewriting \eqref{eq:code}, we have:
    \begin{align}
    \lambda^{t((p-1)\bar {s}+j_p)}&\sum_{i=1}^u\lambda^{t(i-1)\bar {s}\bar n}c_{(p-1)u+i,j}\nonumber\\
    =
    -&\sum_{\begin{substack} {e=1\\e\ne p}\end{substack}}^{\bar n} \lambda^{t((e-1)\bar {s}+j_e)}\sum_{i=1}^u\lambda^{t(i-1)\bar {s}\bar n}c_{(e-1)u+i,j}
    \end{align}
    for all $t=0,\dots,r-1; j=0,\dots,l-1.$
 We will use a subset of the parity-check equations with indices $t$ of the form
$t=wu:$
    \begin{align}
    \lambda^{((p-1)\bar {s}+j_p) wu}\sum_{i=1}^nc_{(p-1)u+i,j}
    =-\sum_{e\ne p}
     \lambda^{((e-1)\bar {s}+j_e) wu}\sum_{i=1}^uc_{(e-1)u+i,j}
    \end{align}
for all $j=0,\dots,l-1;w=0,1,\dots,\bar r-1,$ where we have used the fact that $\lambda^{\bar {s} n}=1.$
Denoting $\alpha=\lambda^u$ and summing these equations on $j_p=0,1,\dots,\bar {s}-1,$ we obtain the following set of conditions:
   \begin{align}\label{eq:twj}
   \sum_{j_p=0}^{\bar {s}-1}\alpha^{((p-1)\bar {s}+j_p)w}\sum_{i=1}^u c_{(p-1)u+i,j}
     =-\sum_{e\ne p}\alpha^{((e-1)\bar {s}+j_e)w}\sum_{j_p=0}^{\bar {s}-1}\sum_{i=1}^u c_{(e-1)u+i,j}
  \end{align}
  for all $w=0,1,\dots,\bar r -1$ and all $j_{\bar n},\dots,j_{p+1},j_{p-1},\dots,j_1,$ where each of the $j$'s ranges over
  $\{0,1,\dots,{\bar s}-1\}.$
   Let $\cR=\{q_1,\ldots,q_{\bar{d}}\}$ be the set of helper racks and let $[\bar{n}]\setminus\cR=\{p,p_1,\ldots,p_{\bar{r}-\bar{s}} \}$.
  	Then \eqref{eq:twj} can be written as follows:
  	\begin{align}
  	\sum_{j_p=0}^{\bar {s}-1}\alpha^{((p-1)\bar {s}+j_p)w}\sum_{i=1}^u c_{(p-1)u+i,j}
  	&+
  	\sum_{\begin{substack}{a\in[\bar{n}]\setminus\cR\\a\neq p}
  		\end{substack}} \alpha^{((a-1)\bar {s}+j_a)w}
  	\sum_{j_p=0}^{\bar {s}-1}\sum_{i=1}^u c_{(a-1)u+i,j}\nonumber\\
  	&=-\sum_{b\in\cR} \alpha^{((b-1)\bar {s}+j_b)w}\sum_{j_p=0}^{\bar {s}-1}\sum_{i=1}^u c_{(b-1)u+i,j}.
  	\end{align} 
In matrix form these equations are shown in \eqref{eq:repairM} above,
where 
   $$
   \zz{\sigma_{e,j(p,0)}}:=\sum_{j_p=0}^{\bar s-1}\sum_{i=1}^u c_{(e-1)u+i,j}, \quad e=1,\dots,\bar n,
   $$
and {$j$ is} as given above after \eqref{eq:twj}.
 
 \begin{figure*}{\small 
  \begin{multline}\hspace*{-.3in}\left[ \begin{array}{*{6}{@{\hspace*{.025in}}c}}
  1&\dots&1&1&\dots&1\\
  \alpha^{\bar {s}(p-1)}&\dots&\alpha^{\bar {s}(p-1)+\bar {s}-1}&\alpha^{\bar {s}(p_1-1)+j_{p_1}}&\dots&\alpha^{\bar {s}(p_{\bar{r}-\bar{s}}-1)+j_{ {p_{\bar{r}-\bar{s} }} } }\\
  \vdots&\vdots&\vdots&\vdots&\vdots&\vdots\\
  (\alpha^{\bar {s}(p-1)})^{\bar r-1}&\dots&(\alpha^{\bar {s}(p-1)+\bar {s}-1})^{\bar r-1}&(\alpha^{\bar {s}(p_1-1)+j_{p_1} })^{\bar r-1}&\dots&(\alpha^{\bar {s}(p_{\bar{r}-\bar{s}}-1)+j_{p_{\bar{r}-\bar{s}}}})^{\bar r-1}
    \end{array}\right]
    \left[\begin{array}{c}\sum_{i=1}^u c_{(p-1)u+i,j(p,0)}\\\vdots\\\sum_{i=1}^u c_{(p-1)u+i,j(p,\bar s-1)}\\
    \sigma_{p_1,\zz{j(p,0)}}\\
    \vdots\\
    \sigma_{p_{\bar{r}-\bar{s}},\zz{j(p,0)}}
    \end{array}\right]\\
    =-\left[\begin{array}{*{3}{@{\hspace*{.025in}}c}}
     1&\dots&1\\
     \alpha^{\bar s (q_1-1) +j_{q_1}}&\dots
     &\alpha^{\bar s(q_{\bar{d}} -1)+j_{q_{\bar{d}}}}\\
     \vdots&\vdots&\vdots\\
     \alpha^{(\bar s (q_1-1) +j_{q_1})(\bar r-1)}&\dots
     &\alpha^{(\bar s(q_{\bar{d}} -1)+j_{q_{\bar{d}}})(\bar r-1)}
     \end{array}\right]
     \left[\begin{array}{c}\sigma_{q_1,\zz{j(p,0)}}\\\vdots\\\sigma_{q_{\bar{d}},\zz{j(p,0)}}
     \end{array}\right]\label{eq:repairM}
  \end{multline}}
  \end{figure*}
 
We claim that Equations \eqref{eq:repairM} suffice to recover one failed node in rack $p$. Indeed, suppose that the {$\bar d$}-dimensional vector on the right-hand side of \eqref{eq:repairM} is made available to the failed node by transmitting one symbol
of $F$ from each of the helper racks. Let us check that the matrix on the left-hand side is Vandermonde, i.e., 
that the defining elements in the second row are distinct. To see this, note that $\text{ord}(\alpha)=\bar {s}\bar n,$ and the
maximum degree of $\alpha$ in the set $\{\alpha^{\bar {s}({e}-1)+m}, m=0,\dots,\bar {s}-1{{;a=1,\ldots,\bar{n}}}\}$ is
  $$
  \bar {s}(\bar n-1)+\bar {s}-1<\bar {s}\bar n.
  $$
  Moreover, {the first $\bar{s}$ coordinates} of the multiplier vector on the left-hand side of \eqref{eq:repairM}
     $$
     \Big(\sum_{i=1}^u c_{(p-1)u+i,j(p,0)},\dots,\sum_{i=1}^u c_{(p-1)u+i,j(p,\bar r-1)}\Big)^T
     $$
contain only one unknown term which corresponds to the failed node. Thus, if the values $c_{(p-1)u+i,j(p,0)}$
of all the functional local nodes are made available to the failed node (recall that this does
not count toward the repair bandwidth), then system \eqref{eq:repairM} can be solved to find the entries of the missing
node. This calculation is repeated $\bar {s}^{\bar n-1}$ times for each assignment of the values 
$j_{\bar n},\dots,j_{p+1},j_{p-1},\dots,j_1,$ thereby completing the repair procedure.
    
Let us compute the inter-rack repair bandwidth of the described procedure. To repair the entries of the single failed node in the $p$th rack with
indices in the subset $\{j(p,a), a=0,1,\dots,\bar {s}-1\}$ we download one symbol of $F$ from each of the ${\bar d}$ helper racks.
There are $\bar {s}^{\bar n-1}$ subsets of the above form, and thus the total repair bandwidth is
   $$
   \bar d\bar s^{\bar n-1}=\frac{\bar{d}l}{\bar s},
   $$ 
proving the optimality claim of the code according to \eqref{eq:rack-bound}.

Finally let us prove that the code $\cC$ is MDS. This is immediate upon observing that each subset of parity-check equations 
isolated by fixing the value of $j=0,1,\dots,l-1$ defines an MDS code. To check this, observe that the 
set of rows {of} the parity-check matrix of $\cC$ for a fixed value $j=(j_{\bar n},\dots,j_1)$ forms a set of parities
of a generalized Reed-Solomon codes (i.e., each column is a set of powers of an element of $F$), and the
defining row of this set of parities is
\begin{align}
 |\lambda^{j_1}, \lambda^{j_1+\bar {s}\bar n},\dots, \lambda^{j_1+(u-1)\bar {s}\bar n}|\lambda^{j_2+\bar {s}},
 &\lambda^{j_2+\bar {s} (1+\bar n)},\dots,
 \lambda^{j_2+\bar {s}(1+(u-1)\bar n)}|\nonumber\\
 &\dots
 |\lambda^{j_{\bar n}+\bar {s}(\bar n-1)},\lambda^{j_{\bar n}+\bar {s}(2\bar n-1)},\dots,
 \lambda^{j_{\bar n}+\bar {s}(\bar n-1+(u-1)\bar n)}|
 \label{eq:row}
 \end{align}
 where each group between the vertical bars corresponds to a fixed value of $s=1,\dots,\bar n$ in \eqref{eq:code}.
 It suffices to show that all these elements are distinct or that
these groups do not overlap. Note that the largest power in \eqref{eq:row} is
  \begin{equation}\label{eq:max}
  j_{\bar n}+\bar {s}(\bar n-1+(u-1)\bar n)\le \bar {s}-1+u\bar n \bar {s}-\bar {s}<\bar {s}n=\text{ord}(\lambda).
  \end{equation}
  Now consider two groups and let their numbers be $a$ and $b$, where $1\le {b<a}\le\bar n$.
Then the difference between the exponents of the first elements in the {two} groups is
  $$
  (a-b)\bar {s}+(j_a-j_b)\ge 1
  $$
so the first elements are obviously distinct. Further, the exponents of the elements in each of the groups are obtained by adding
a multiple of $\bar {s} \bar n$ to the exponent of the first element, which together with \eqref{eq:max} implies that the groups are disjoint.
This shows that the code $\cC$ is MDS, and the proof is complete.
\end{IEEEproof}

We remark that the repair procedure relies on a subset of the parity-check equations of the code $\cC$. Namely, the only
rows of the parity-check matrix that we use are the rows whose numbers are integer multiples of the size of the rack $u$. It suffices to use
only these parities because the assumptions of the rack model are relaxed compared to the standard definition of regenerating codes.
The remaining parities support the MDS property of the code $\cC$ and do not contribute to the repair procedure. 

In Sec. \ref{sec:OArack} we construct codes with somewhat better parameters than the codes given by Construction~\ref{Construction1}.
Specifically, the smallest field size required for the code family in Sec.~\ref{sec:OArack} is $n+\bar s-1$ (as opposed to $\bar s n$), and the repair procedure
accesses fewer symbols on the helper nodes than the procedure presented in the above proof.
At the same time, the codes presented in this section have the {\em optimal update property.} Namely, a codeword of the code $\cC$ 
can be viewed as an $l\times n$ array, and for a given row index $j\in\{1,\dots,l-1\}$ the $n$ symbols are encoded with a generalized RS code 
independently of the other rows. Thus, if some $k$ symbols are taken as information symbols, then the change of one symbol in the data 
requires to change $r$ parity symbols, which is also the smallest possible number \cite{Tamo14}. At the same time, the codes in the family
of Sec.~\ref{sec:OArack} do not have optimal update, and are in this respect inferior to the present construction.

\section{Low-access codes for the rack model}
This section aims at constructing an optimal-repair MSR code for the rack model that accesses a reduced number of symbols on the nodes
in the helper racks. Our presentation is formed of two parts. In the first part we
construct an optimal-access MSR code for arbitrary repair degree $k\le d\le n-1$ {\em without assuming the rack model} of storage. The code
has subpacketization $l=(d-k+1)^n$. In the second part we present a modification of this construction for the rack model, attaining
subpacketization $l=\bar s^{\bar n}.$ Note that this value is smaller than the smallest node size of known constructions of
OA codes for the homogeneous model, which is $s^n$ \cite{ye2017explicit}.

	\subsection{Optimal-access MSR codes with arbitrary repair degree for homogeneous storage}\label{sec:OAcode}
In this section we present a family of OA codes for any repair degree $k\le d\le n-1.$
Let $s=d-k+1$ and let $F,|F|\geq n+s-1$ be a finite field. 
Let $\lambda_0,\ldots,\lambda_{n-1},\mu_1,\ldots,\mu_{s-1}$ be $n+s-1$ distinct elements of $F$. 
Let $i=(i_{n-1},\ldots,i_0)$ be the $s$-ary representation of $i=0,\ldots,l-1$ and (as before) let $i(a,b)=(i_{n-1},\ldots,i_{a+1},b,i_{a-1},\ldots,i_0)$ for $0\le a \le n-1$ and $0\le b \le s-1$. For brevity below we use the notation $$\delta(i):=\mathbbm{1}_{\{i=0\}}.$$
	
\begin{construction} {\rm Define an $(n,k=n-r,l=s^n)$ array code $\cC=\{\bfc=(c_{j,i})_{0\le j\le n-1; 0\le i\le l-1}\}$, where the codeword
$\bfc$ satisfies the following parity check equations over $F$:
	\begin{align}
		\sum_{j=0}^{n-1}\lambda_j^t c_{j,i} + 
		\sum_{j=0}^{n-1}\delta(i_j)\sum_{p=1}^{s-1}\mu_p^t c_{j,i(j,p)}=0 ,\quad i=0,\ldots,l-1;\,t=0,\ldots,r-1.
				\label{eq:oa-msr-pc}
	\end{align}
Since later in this section we rely on multiplicative structure of $F$, we label the nodes $0,\dots,n-1$ and not $1,\dots,n$ as in Construction \ref{Construction1}.
In the next subsection we will also label the racks from $0$ to $\bar n-1$ for the same reason.}
\end{construction}
	
\begin{theorem} The code $\cC$ defined in \eqref{eq:oa-msr-pc} is an optimal-access MDS array code.
\end{theorem}
\begin{IEEEproof}
	{\sc I. Optimal-access property.}
Let $j_1\in \{0,\dots,n-1\}$ and suppose that $c_{j_1}$ is the failed node. Choose any set $\cR$ of $d$ nodes which will be the helper nodes in the repair procedure. 
	
Let $\cJ=\cR^c$ be the complement of $\cR$ in the set of all nodes. Choose $a, 1\le a\le n-d$ and let $\zz{\cJ_{a}}, |\zz{\cJ_{a}}|=a, j_1\in\zz{\cJ_{a}}$ be any subset of the set $\cJ.$ In particular, $\cJ_1=\{j_1\},$
and there are $\binom{n-d-1}{a-1}$ possible choices for $\zz{\cJ_{a}}, a\ge 2$ (overloading the symbol $a,$ we use it both as the size of the
subsets and the label of the subsets of size $a$).

Let $\cI \subset \{0,\dots,l-1\}$ be the set of indices such that $i_{j_1}=0$ and define a subset $\cI_a\subset\cI$ as follows:
    \begin{gather*}
    \cI_a=\bigcup_{\zz{\cJ_{a}}\subseteq\cJ}\cI(\zz{\cJ_{a}}),\quad a=1,\dots,n-d
    \end{gather*}
where
    \begin{gather*}
    \cI(\zz{\cJ_{a}}):=\{i\in\{0,\dots,l-1\}|i_j=0,j\in\zz{\cJ_{a}}; i_{j'} \neq 0, j'\in\cJ\setminus\zz{\cJ_{a}}\}.
    \end{gather*}
It is easy to see that the sets $\cI_{\zz{a}}$ partition the set $\cI:$ Indeed, these sets are clearly disjoint,
$|\cI(\zz{\cJ_{a}})|=(s-1)^{n-d-a}s^{n-(n-d)}$ for any choice of $\zz{\cJ_{a}},$ and
   $$
   \sum_{a=1}^{n-d}\binom{n-d-1}{a-1}(s-1)^{n-d-a}s^d=s^{n-1}=|\cI|.
   $$
	
The repair of the node $c_{j_1}$ will be accomplished using the parity-check equations that correspond to $i\in \cI$ and all $t$ in 
\eqref{eq:oa-msr-pc}.
We will use induction on $a.$ 

Let us start with the case $a=1.$ Our aim is to show that it is possible to find the values
   \begin{align}
    &c_{j_1,i(j_1,p)},\quad p=0,\ldots,s-1 \label{eq:j1}\\
    &c_{j_w,i},\quad w=2,\ldots,n-d; i\in\cI_1\label{eq:w}
    \end{align}
from the helper nodes $\{c_j\mid j\in \cR \}$. By definition of the set $\cI_1$ we have $\delta(i_j)=0$ for $i\in \cJ\backslash \cJ_1$.
On account of this, from \eqref{eq:oa-msr-pc}, for $i\in\cI_1$, we have 
	\begin{align}
		\sum_{j\in\cJ}\lambda_{j}^t c_{j,i}
		+
		\sum_{p=1}^{s-1}\mu_{p}^t c_{j_1,i(j_1,p)}
		&= - 
		\sum_{j\in\cR}\Big(\lambda_{j}^t c_{j,i}
		+
				\delta(i_j)
		\sum_{p=1}^{s-1}\mu_{p}^t c_{j,i(j,p)}\Big)
		\label{eq:oa-msr-pc-i}
	\end{align}
for all $t=0,1,\dots,r-1$. For simplicity let us denote the right-hand side of \eqref{eq:oa-msr-pc-i} by $\sigma_{i,t}(\cJ_1).$
Writing Equations \eqref{eq:oa-msr-pc-i} in matrix form and reordering the variables to match \eqref{eq:j1}-\eqref{eq:w}, 
we obtain
	\begin{align}
		\begin{bmatrix}
		1 & 1 & \cdots & 1 & 1 & \cdots  & 1 \\
		\lambda_{j_1} & \mu_{1} & \cdots & \mu_{s-1} & \lambda_{j_2} & \cdots  & \lambda_{j_{n-d}} \\
		\vdots  & \vdots & \ddots & \vdots & \vdots & \ddots & \vdots \\
		\lambda_{j_1}^{r-1} & \mu_{1}^{r-1} & \cdots & \mu_{s-1}^{r-1} & \lambda_{j_2}^{r-1} & \cdots  & \lambda_{j_{n-d}}^{r-1} 
		\end{bmatrix}
		\begin{bmatrix}
		c_{j_1,i(j_1,0)} \\
		c_{j_1,i(j_1,1)}\\
		\vdots \\
		c_{j_1,i(u_1,s-1)} \\
		c_{j_2,i} \\
		\vdots \\
		c_{j_{n-d},i}
		\end{bmatrix} =
		\begin{bmatrix}
		\sigma_{i,0}(\cJ_1) \\
		\sigma_{i,1}(\cJ_1) \\
		\vdots  \\
		\sigma_{i,r-1}(\cJ_1)
		\end{bmatrix}.\label{eq:matrix}
	\end{align}
Observe that the matrix in \eqref{eq:matrix} is invertible. Therefore, the values listed in \eqref{eq:j1}--\eqref{eq:w} can be found from the values $\{ \sigma_{i,t}(\cJ(1))\mid t = 0, \ldots, r-1\}$ for every $i\in\cI(1)$. This completes the proof of the induction basis.
	
	Now suppose that we have recovered the values $\{c_{j_1,i(j_1,p)}\mid p=0,\ldots,s-1 \}$ and $\{c_{j_w,i}\mid w=2,\ldots,n-d\}$ for every $i\in\zz{\cI_{a'}}$ and $1\le a' \le a-1$, where $2\le a\le n-d$. 
	Let us make the induction step. We begin with fixing some subset $\zz{\cJ_{a}\subseteq\cJ}$.
	Let $i\in \cI(\zz{\cJ_{a}})$. 
	From \eqref{eq:oa-msr-pc}, we have
	\begin{align}
		\sum_{j\in\cJ}\lambda_{j}^t c_{j,i}
		+
		\sum_{j\in\zz{\cJ_{a}}}
		\sum_{p=1}^{s-1}\mu_{p}^t c_{j,i(j,p)}
		& =
		\sum_{j\in\cJ}\lambda_{j}^t c_{j,i}
		+
		\sum_{p=1}^{s-1}\mu_p^t 
		\sum_{j\in\zz{\cJ_{a}}} c_{j,i(j,p)} 
		\nonumber\\
		& = - 
		\sum_{i\in \cR}\Big(\lambda_{j}^t c_{j,i}
		+
				\delta(i_j) 
		\sum_{p=1}^{s-1}\mu_{p}^t c_{j,i(j,p)}\Big)
		\label{eq:oa-msr-pc-iua}
	\end{align}
for all $t=0,1,\dots,r-1$.	Let us denote the right-hand side of \eqref{eq:oa-msr-pc-iua} by $\sigma_{i,t}(\zz{\cJ_{a}})$ and let
	\begin{align*}
		\rho_{i,p} &: = \sum_{j\in\zz{\cJ_{a}}} c_{j,i(j,p)}.
			\end{align*}
	As before, the value of $\sigma_{i,t}(\zz{\cJ_{a}})$ depends only on the helper nodes $\{c_j\mid j\in \cR\}$. Writing equations
\eqref{eq:oa-msr-pc-iua} in matrix form, we obtain
	\begin{align}
		\begin{bmatrix}
		1 & \cdots & 1 & 1 & \cdots  & 1 \\
		\mu_1 & \cdots & \mu_{s-1} & \lambda_{j_1} & \cdots  & \lambda_{j_{n-d}} \\
		\vdots & \ddots & \vdots & \vdots & \ddots & \vdots \\
		\mu_1^{r-1} & \cdots & \mu_{s-1}^{r-1} & \lambda_{j_1}^{r-1} & \cdots  & \lambda_{j_{n-d}}^{r-1} 
		\end{bmatrix}
		\begin{bmatrix}
		\rho_{i,1} \\
		\vdots \\
		\rho_{i,s-1} \\
		c_{j_1,i} \\
		\vdots \\
		c_{j_{n-d},i}
		\end{bmatrix} =
		\begin{bmatrix}
		\sigma_{i,0}(\zz{\cJ_{a}}) \\
		\sigma_{i,1}(\zz{\cJ_{a}}) \\
		\vdots  \\
		\sigma_{i,r-1}(\zz{\cJ_{a}})
		\end{bmatrix}.
		\label{eq:oa-msr-mx-iua}
	\end{align} 
	Therefore, for any $\zz{\cJ_{a}}\subseteq\cJ$ and every $i\in\cI(\zz{\cJ_{a}})$, the values $\{\rho_{i,p}\mid p = 1,\ldots,s-1 \}$ and $\{c_{j_w,i}\mid w=1,\ldots,n-d \}$ can be calculated from the values $\{\sigma_{i,t}(\zz{\cJ_{a}})\mid t = 0, \ldots, r-1 \}$. Since there are no assumptions
on the choice of $\cJ_a,$ this ensures that we can find the values $\{\rho_{i,p}\mid p = 1,\ldots,s-1 \}$ and $\{c_{j_w,i}\mid w=1,\ldots,n-d \}$ for all $i\in\zz{\cI_{a}}$.

	Note that for $i\in\cI(\zz{\cJ_{a}})$, $j \in \zz{\cJ_{a}}\setminus\{j_1\}$, and $p \neq 0$, we have $i(j,p) \in \zz{\cI_{a-1}}$. By the induction hypothesis, we have recovered the values $\{c_{j_w,i} \mid i\in\zz{\cI_{a-1}},w=2,\ldots,n-d  \}$, and therefore, we know the values $\{c_{j,i(j,p)} \mid j\in\zz{\cJ_{a}}\setminus\{j_1\},p\neq 0\}$ for each $i\in\cI_{\zz{a}}$. 
	With these values and $\{\rho_{i,p} \mid i\in\cI_{\zz{a}},p=1,\ldots,s-1 \}$, we can obtain the values $\{c_{j_1,i(j_1,p)}\mid p=1,\ldots,s-1 \}$.
	
	Thus, overall we can recover the values $\{c_{j_1,i(j_1,p)}\mid i\in\cI,p=0,\ldots,s-1 \}=\{c_{j_1,i}\mid i=0,\ldots,l-1\}$ from the helper nodes $\{c_j \mid j\notin\cJ \}$.
	
	Now let us count the number of symbols we access in each helper node. It is clear from the definition of $\zz{\sigma_{i,t}(\cJ_{a})}$ that we need to access the symbols $\{c_{j,i}\mid i\in\cI\}$ for each $j\in \cR$. Hence, the number of symbols we access to repair $c_{j_1}$ is 
	\begin{align}
		d|\cI|=ds^{n-1}=\frac{dl}{d-k+1}.
	\end{align}
	In other words, the number of symbols we access meets the cut-set bound \cite{dimakis2010network} for the repair bandwidth. Furthermore, observe that the set of symbols we access in each helper node depends on index of failed node but not the index of the helper node.
Indeed, even though the entries of the helper nodes involved in \eqref{eq:oa-msr-pc-iua} depend on $j\in\cR,$ overall the entries
of the helper nodes accessed are indexed by the values $i\in \cI$ (since the sets $\cI_a$ partition the set $\cI$).
 Thus, the repair matrices can be defined independently of the choice of the subset $\zz{\cR\subseteq\{0,\dots,n-1\}}.$
	
\vspace*{.1in}	{\sc II. MDS property.}
We will show that any $r$ nodes can be recovered from the other $k$ nodes of the codeword.

Let $\cJ=\{j_1,\ldots,j_r\}\subseteq\{0,\ldots,n-1\}$ be a set of $r$ nodes. 
	As before, we will denote $a$-subsets of $\cJ$ by $\zz{\cJ_{a}}$, $\zz{0\le a\le r}$.
		Let $\cI_0=\{ i=(i_{n-1},\ldots,i_0)\in\{0,1,\dots,l-1\} \mid i_{j} \neq 0, j\in\cJ \}$ and let
\begin{gather*}
		\cI(\zz{\cJ_{a}})=\{ i=(i_{n-1},\ldots,i_0)\in\{0,1,\dots,l-1\} \mid i_j=0,j\in\zz{\cJ_{a}}; i_{j'} \neq 0, j'\in\cJ\setminus\zz{\cJ_{a}} \},\\
		\hspace*{4in}{\zz{\cJ_{a}}\subseteq\cJ, 1 \le a \le r},\\
		\cI_{a}=\bigcup_{\zz{\cJ_{a}}\subseteq\cJ}\cI(\zz{\cJ_{a}}).
\end{gather*}
Observe that the sets $\cI_a, 0\le a\le r$ partition the set $\{0,1,\dots,l-1\}.$
	
To prove the MDS property we use induction on $a.$
For the induction basis, we will show that it is possible to recover the values $\{c_{j,i}\mid j\in\cJ ,i\in\cI_0\}$ from the nodes $\{c_j \mid j\in \cJ^c \}$. From \eqref{eq:oa-msr-pc}, for $i\in\cI_0$, we have
	\begin{align}
		\sum_{j\in\cJ}\lambda_{j}^t c_{j,i}
		&= - 
		\sum_{j\in \zz{\cJ^{c}}}\Big(
		\zz{\lambda_{j}^t c_{j,i}
		+
				\delta(i_j)}\sum_{p=1}^{s-1}\mu_{p}^t c_{j,i(j,p)}
			\Big), \quad t=0,\ldots,r-1.
		\label{eq:mds}
	\end{align}
To simplify notation, denote the right-hand side of \eqref{eq:mds} by $ \sigma_{i,t}=\sigma_{i,t}(\emptyset).$	
	Note that the value of $\sigma_{i,t}$ depends only on the nodes $\{c_j\mid j\in \cJ^c\}$. Writing \eqref{eq:mds} 
in matrix form, we obtain
	\begin{align*}
		\begin{bmatrix}
		1 & \cdots  & 1 \\
		\lambda_{j_1} & \cdots  & \lambda_{j_{r}} \\
		\vdots & \ddots & \vdots \\
		\lambda_{j_1}^{r-1} & \cdots  & \lambda_{j_{r}}^{r-1} 
		\end{bmatrix}
		\begin{bmatrix}
		c_{j_1,i} \\
		c_{j_2,i} \\
		\vdots \\
		c_{j_r,i}
		\end{bmatrix} =
		\begin{bmatrix}
		\sigma_{i,0} \\
		\sigma_{i,1} \\
		\vdots  \\
		\sigma_{i,r-1}
		\end{bmatrix}.
	\end{align*}
This equation implies that the values $\{c_{j,i} \mid j\in\cJ \}$ can be calculated from the values $\{ \sigma_{i,t}\mid t = 0, \ldots, r-1\}$ for every $i\in\cI_0$.
	
For the induction step, let $1\le a\le r$ and suppose that it is possible to recover the values $\{c_{j,i}\mid j\in\cJ \}$ for every $i\in\cI_{a'}$ and $0\le a' \le a-1$ from the nodes $\{c_j\mid j\in\zz{\cJ^{c}} \}.$
	
Now let us fix a set $\zz{\cJ_{a}\subseteq \cJ}$ and let $i\in \cI(\zz{\cJ_{a}})$. 
	From \eqref{eq:oa-msr-pc}, we have
	\begin{align}
		\sum_{j\in\cJ}\lambda_{j}^t c_{j,i}
		& =
		-
		\sum_{p=1}^{s-1}\mu_p^t 
		\sum_{j\in\zz{\cJ_{a}}} c_{j,i(j,p)} 
		 - 
		\sum_{j\in\zz{\cJ^{c}}}\Big(
		\zz{\lambda_{j}^t c_{j,i}
		+
				\delta(i_j)}\sum_{p=1}^{s-1}\mu_{p}^t c_{j,i(j,p)}\Big)\nonumber\\
				&=:-\rho'_{i,t}- \sigma_{i,t}(\zz{\cJ_{a}})
		\label{eq:mds-iua}
	\end{align}
$t=0,1,\dots,r-1,$ where the second line serves to introduce the shorthand notation.
We argue that both quantities $\rho'$ and $\sigma$ can be found from the nodes outside the set of the chosen $r$ nodes, i.e., from
$\{c_j\mid j\in\cJ^c\}.$ This claim is obvious for the $\sigma$'s and
constitutes the induction hypothesis for the $(\rho')$'s. Indeed, 
for $i\in\cI(\zz{\cJ_{a}})$, $j \in \zz{\cJ_{a}}$ and $p \neq 0$, we have $i(j,p) \in \cI_{a-1}$. 
By the induction hypothesis, we have recovered the values $\{c_{j,i} \mid i\in\cI_{a-1}, j\in\cJ \}$, 
and therefore, we know the values $\{c_{j,i(j,p)} \mid j\in\zz{\cJ_{a}},p\neq 0 \}$ for each $i\in\cI_{\zz{a}}$. 
Writing relations \eqref{eq:mds-iua} for all $t=0,\ldots,r-1$ in matrix form, we obtain
	\begin{align}
		\begin{bmatrix}
		1 & \cdots  & 1 \\
		\lambda_{j_1} & \cdots  & \lambda_{j_{r}} \\
		\vdots & \ddots & \vdots \\
		\lambda_{j_1}^{r-1} & \cdots  & \lambda_{j_{r}}^{r-1} 
		\end{bmatrix}
		\begin{bmatrix}
		c_{j_1,i} \\
		c_{j_2,i} \\
		\vdots \\
		c_{j_r,i}
		\end{bmatrix} =
		\begin{bmatrix}
		\rho'_{i,0}+\sigma_{i,0}(\zz{\cJ_{a}}) \\
		\rho'_{i,1}+\sigma_{i,1}(\zz{\cJ_{a}}) \\
		\vdots  \\
		\rho'_{i,r-1}+\sigma_{i,r-1}(\zz{\cJ_{a}})
		\end{bmatrix}.
	\end{align}
This establishes the induction step. Therefore, it is possible to find the values $\{c_{j,i} \mid j\in\cJ \}$ for every $i\in\cI(\zz{\cJ_{a}})$ and $\zz{\cJ_{a}}\subseteq\cJ$. It follows that we can recover the values $\{c_{j,i}\mid j\in\cJ \}$ for every $i\in\cI_{a}$ and all $a=0,1,\dots,r.$ 
Since the sets $\zz{I_a, 0\le a\le r}$ form a partition the set $\{0,1,\dots,l-1\}$, we have shown that all the values $\{c_{j,i} \mid j\in\cJ,i\in\cI_{a},0\le a\le r\}=\{c_{j,i}\mid j\in\cJ,i\in\{0,1,\dots,l-1\}\}$ can be recovered from the available nodes $\{c_j\mid j\in\cJ^c\}.$ Thus, any $r$ nodes in the
codeword can be found from the complementary set of $k$ nodes, which proves the MDS property.
	\end{IEEEproof}

	\subsection{Rack-aware MSR codes with low access}\label{sec:OArack}
In this section we adapt the code family constructed in Sec.~\ref{sec:OAcode} for the rack-aware storage model.
This result is obtained by adjusting the sub-packetization and by carefully choosing the elements $\lambda_0,\ldots,\lambda_{n-1}.$  

We aim to construct an $(n,k,l)$ MDS array code over $F$, where $n=\bar n u,$ and $u$ is the size of the \zz{rack}.
Recall that $\bar{s}=\bar{d}-\bar{k}+1$ where $\bar{k}\le \bar{d}\le \bar{n}-1,$ and $\bar{k}=\lfloor k/u \rfloor$. Let $|F|\geq n+\bar{s}-1$ and $n|(|F|-1).$ Let $\lambda\in F$ be an element of multiplicative order $n=\bar{n}u$, and let $\mu_1,\ldots,\mu_{\bar{s}-1}$ be $\bar{s}-1$ distinct elements in $F\setminus\{\lambda^i\mid i=0,\ldots,n-1 \}$. For $j=0,\ldots,n-1$, let us write $j={e}u+{g}$ where $0\le {e}< \bar{n}$ and $0\le {g}< u.$

We construct an  rack-aware low-access MSR code over $F$ that can repair any single node from any $\bar{d}$ helper racks.

\begin{construction}\label{OAconstruction} Define an $(n,k=n-r,l=\bar{s}^{\bar{n}})$ array code $\cC=\{(c_{j,i})_{0\le j\le n-1; 0\le i\le l-1}\}$ by the following parity-check equations over $F$:
	\begin{align}
		\sum_{j=0}^{n-1}\lambda_j^t c_{j,i} + 
		\sum_{j=0}^{n - 1} 
		\zz{\delta(i_e)} 
		\sum_{p=1}^{\bar{s}-1}\mu_p^t c_{j,i({e},p)}
		=0,\label{eq:oa-msr-pc-rack}
	\end{align}
	where $\lambda_j = \lambda^{{e}+{g}\bar{n}}$, $i=0,\ldots,l-1$ and $t=0,\ldots,r-1.$
\end{construction}

We will show that this code family supports optimal repair while accessing $l/{\bar s}$ symbols on each of the nodes in the helper racks, which
is by a factor of $s/\bar s\approx u$ greater than the bound in Prop.~\ref{prop:optimal-access}. While these codes stop short of attaining the bound \eqref{eq:bound-access}, they have
lower access requirement than the codes given by Construction~\ref{Construction1}, which access all symbols of the helper nodes, i.e., $\bar s$ times 
more symbols than the current construction.

\begin{theorem}
The code $\cC$ defined in \eqref{eq:oa-msr-pc-rack} is an optimal-repair MDS array code. The repair procedure accesses $l/{\bar s}$ symbols on each
of the nodes in $\bar d$ helper racks. The repair scheme does not depend on the choice of the subset of $\bar d$ helper racks.
\end{theorem}

\begin{IEEEproof}
	\emph{Optimal repair property.}
	
	Suppose $c_{j_1}$ is the failed node, where $j_1={e}_1 u+{g}_1$. Let $\cR$ be the set of helper racks and let $\cJ=\zz{\{0,\ldots,\bar{n}-1\}}\backslash \cR.$
We write this set as $\cJ=\{{e}_1,{e}_2,\ldots,{e}_{\bar{n}-\bar{d}}\}.$ 
For a given $a, 1 \le a \le \bar{n}-\bar{d}$ we will need $a$-subsets of $\cJ$, which we denote by $\zz{\cJ_{a}}.$ We always assume that ${e}_1\in \zz{\cJ_{a}}.$ As before, let $\cI\subset \{0,1,\dots,l-1\}$ be the subset of indices such that $i_{{e}_1}=0$; let
  $$
  \cI_1=\{ i=(i_{\bar{n}-1},\ldots,i_0)\in\zz{\{0,\ldots,l-1\}} \mid i_{{e}_1} = 0; i_{{e}} \neq 0, {e}\in{\cJ}\setminus{\cJ}_1 \}
  $$ 
  and define 
    $$
    \cI_a=\bigcup_{\zz{\cJ_a}\subseteq{\cJ}}\cI(\zz{\cJ_a}), \quad a=2,\dots,\bar n-\bar d,
  $$
where
    $$
    \cI(\zz{\cJ_a})=\{ i=(i_{\bar{n}-1},\ldots,i_0)\in\{0,\dots,l-1\} \mid i_{e}=0,{e}\in\zz{\cJ_a}; i_{{e}} \neq 0, {e}\in{\cJ}\setminus\zz{\cJ_a} \}.
    $$
%
%
	
Recall that $\bar{r}=\bar{n}-\bar{k}$. We will use the parity check equations corresponding to $i\in\cI$ and all powers $t=uw,w=0,\ldots,\bar{r}-1$ to repair $c_{j_1}$. To show that the repair is possible,  we argue by induction on $a=1,\ldots,\bar{n}-\bar{d}$.
	
To prove the induction basis, we show that it is possible to recover the values 
$\{c_{j_1,i({e}_1,p)}\mid p=0,\ldots,\bar{s}-1 \}$ and $\{\sum_{{g}=0}^{u-1}c_{{e} u+{g},i}\mid {e}\in{\cJ}\setminus{\cJ}_1\}$ for every $i\in\cI_1$ from the helper racks $\zz{\cR}$. 
	From \eqref{eq:oa-msr-pc-rack}, for $i\in\cI_1$, we have
	\begin{align}
		\sum_{{e}\in{\cJ}}\sum_{{g}=0}^{u-1}\lambda_{{e} u+{g}}^t c_{{e} u+{g},i}
		+
		&\sum_{{g}=0}^{u-1}\sum_{p=1}^{\bar{s}-1}\mu_{p}^t c_{{e}_1 u+{g},i({e}_1,p)}\nonumber\\
		&= - 
		\sum_{{e}\in\cR}\sum_{g=0}^{u-1}\Big(
		\zz{\lambda_{{e}u+{g}}^t c_{{e}u+{g},i}
		+
		\delta(i_e)}
		\sum_{p=1}^{\bar{s}-1}\mu_{p}^t c_{{e}u+{g},i({e},p)}\Big).
	\end{align}
	
	Using $t=uw$, $\lambda_{{e}u+{g}}=\lambda^{{e}+{g}\bar{n}}$, and $\lambda^{\bar{n}u}=1$, we obtain
	\begin{align}
		\sum_{{e}\in{\cJ}}\lambda^{{e} uw} \sum_{{g}=0}^{u-1}c_{{e} u+{g},i}
		&+
		\sum_{p=1}^{\bar{s}-1}\mu_{p}^{uw} \sum_{{g}=0}^{u-1}c_{{e}_1 u+{g},i({e}_1,p)}
		\nonumber\\
		&= - 
		\sum_{{e}\in\cR}\Big(
		\lambda^{{e}uw} \sum_{{g}=0}^{u-1}\zz{c_{{e}u+{g},i}
		+
		\delta(i_e)}
		\sum_{p=1}^{\bar{s}-1}\mu_{p}^{uw} 
		\sum_{{g}=0}^{u-1}c_{{e}u+{g},i({e},p)}\Big),
		\label{eq:oa-msr-pc-i-rack}
	\end{align}
	$i\in \cI_1,w=0,\dots,\zz{\bar{r}-1}.$
To shorten our notation, denote the right-hand side of \eqref{eq:oa-msr-pc-i-rack} by $\sigma_{i,w}({\cJ}_1)$ and let
	\begin{align*}
		\pi_{i,{e}} &:= \sum_{{g}=0}^{u-1}c_{{e} u+{g},i}.
	\end{align*}
%
Note that the value of $\sigma_{i,w}({\cJ}_1)$ only depends on the helper racks. For $i=1,\ldots,\bar{n}-\bar{d}$ 
define $\alpha_i:=\lambda^{{e}_iu}.$ Let us write equations \eqref{eq:oa-msr-pc-i-rack} for all $w=0,\ldots,\bar{r}-1$ in matrix form:
	\begin{align}\label{eq:ma}
		\begin{bmatrix}
		1 & 1 & \cdots & 1 & 1 & \cdots  & 1 \\
		\alpha_{1} & \mu_{1} & \cdots & \mu_{\bar{s}-1} & \alpha_{2} & \cdots  & \alpha_{\bar{n}-\bar{d}} \\
		\vdots  & \vdots & \ddots & \vdots & \vdots & \ddots & \vdots \\
		\alpha_{1}^{\bar{r}-1} & \mu_{1}^{\bar{r}-1} & \cdots & \mu_{\bar{s}-1}^{\bar{r}-1} & \alpha_{2}^{\bar{r}-1} & \cdots  & \alpha_{\bar{n}-\bar{d}}^{\bar{r}-1} 
		\end{bmatrix}
		\begin{bmatrix}
		\pi_{i,{e}_1} \\
		\pi_{i({e}_1,1),{e}_1}\\
		\vdots \\
		\pi_{i({e}_1,\bar{s}-1),{e}_1} \\
		\pi_{i,{e}_2} \\
		\vdots \\
		\pi_{i,{e}_{\bar{n}-\bar{d}}}
		\end{bmatrix} =
		\begin{bmatrix}
		\sigma_{i,0}(\cJ_1) \\
		\sigma_{i,1}(\cJ_1) \\
		\vdots  \\
		\sigma_{i,\bar{r}-1}(\cJ_1)
		\end{bmatrix}.
	\end{align}
	Observe that the matrix on the left-hand side of \eqref{eq:ma} is invertible. Therefore, the values $\{c_{j_1,i(j_1,p)} \mid p=0,\ldots,\bar{s}-1\}$ and $\{\sum_{{g}=0}^{u-1}c_{{e} u+{g},i}\mid {e}\in{\cJ}\setminus{\cJ}_1 \}$ can be found from the values 
	$\{\sigma_{i,w}(\cJ_1)\mid w = 0, \ldots, \bar{r}-1\}$ and the local nodes $\{c_{{e}_1 u+{g}} \mid {g}\neq {g}_1 \}$ for every $i\in\cI_1$. This completes the proof of the induction basis.

	Now let us fix $a\in\{2,\dots,\bar n-\bar d\}$ and suppose that we have recovered the values 
	\begin{gather*}
	\{c_{j_1,i({e}_1,p)}\mid p=0,\ldots,\bar{s}-1 \}\text{ and }\Big\{\sum_{{g}=0}^{u-1}c_{{e} u+{g},i} \mid {e}\in{\cJ}\setminus{\cJ}_1\Big\}
	\\i\in\cI_{a'};\ 1\le a' \le a-1
	\end{gather*}
	 from the information downloaded from the helper racks $\zz{\cR}.$
	
Fix a subset $\zz{\cJ_a}, |\zz{\cJ_a}|=a,$ and let $i\in \cI(\zz{\cJ_a})$. 
	From \eqref{eq:oa-msr-pc-rack}, we have
	\begin{align}
		\sum_{{e}\in{\cJ}}\sum_{{g}=0}^{u-1}\lambda_{{e} u+{g}}^t c_{{e} u+{g},i}
		&+
		\sum_{{e}\in\zz{\cJ_a}}\sum_{{g}=0}^{u-1}\sum_{p=1}^{\bar{s}-1}\mu_{p}^t c_{{e} u+{g},i({e},p)}\nonumber\\
		=
		&\sum_{{e}\in{\cJ}}\sum_{{g}=0}^{u-1}\lambda_{{e} u+{g}}^t c_{{e} u+{g},i}
		+
		\sum_{p=1}^{\bar{s}-1}\mu_{p}^t 
		\sum_{{e}\in\zz{\cJ_a}}
		\sum_{{g}=0}^{u-1}
		c_{{e} u+{g},i({e},p)}\nonumber\\
		=  
		&-\sum_{{e}\in\cR}\Big(\sum_{{g}=0}^{u-1}\lambda_{{e}u+{g}}^t \zz{c_{{e}u+{g},i}
		+
		\delta(i_e)}
		\sum_{p=1}^{\bar{s}-1}\mu_{p}^t 
		\sum_{{g}=0}^{u-1} c_{{e}u+{g},i({e},p)}\Big).
	\end{align}
	Using $t=uw$, $\lambda_{{e}u+{g}}=\lambda^{{e}+{g}\bar{n}}$, and $\lambda^{\bar{n}u}=1$, we obtain
	\begin{align}
		\sum_{{e}\in{\cJ}}\lambda^{{e} uw} \sum_{{g}=0}^{u-1}c_{{e} u+{g},i}
		&+
		\sum_{p=1}^{\bar{s}-1}\mu_{p}^{uw} 
		\sum_{{e}\in\zz{\cJ_a}}\sum_{{g}=0}^{u-1}c_{{e}_1 u+{g},i({e}_1,p)}
		\nonumber\\
		&= - 
		\sum_{{e}\in \cR}\Big(\lambda^{{e}uw} \sum_{{g}=0}^{u-1} \zz{c_{{e}u+{g},i}
		+
		\delta(i_e)}
		\sum_{p=1}^{\bar{s}-1}\mu_{p}^{uw} 
		\sum_{{g}=0}^{u-1}c_{{e}u+{g},i({e},p)}\Big).
		\label{eq:oa-msr-pc-iua-rack}
	\end{align}
Again for notational convenience denote the right-hand side of \eqref{eq:oa-msr-pc-iua-rack} by $\sigma_{i,w}(\zz{\cJ_a})$ and let
	\begin{align*}
		\rho_{i,p} & := \sum_{{e}\in\zz{\cJ_a}}\sum_{{g}=0}^{u-1}
		c_{{e} u+{g},i({e},p)}.
		\end{align*}
%
	Note that the value of $\zz{\sigma_{i,w}(\cJ_a)}$ depends only on the information in the helper racks. 
Let us write Equations \eqref{eq:oa-msr-pc-iua-rack} for all $w=0,\ldots,\bar{r}-1$ in matrix form:
	\begin{align}
		\begin{bmatrix}
		1 & \cdots & 1 & 1 & \cdots  & 1 \\
		\mu_1 & \cdots & \mu_{\bar{s}-1} & \alpha_{1} & \cdots  & \alpha_{\bar{n}-\bar{d}} \\
		\vdots & \ddots & \vdots & \vdots & \ddots & \vdots \\
		\mu_1^{\bar{r}-1} & \cdots & \mu_{\bar{s}-1}^{\bar{r}-1} & \alpha_{1}^{\bar{r}-1} & \cdots  & \alpha_{\bar{n}-\bar{d}}^{\bar{r}-1} 
		\end{bmatrix}
		\begin{bmatrix}
		\rho_{i,1} \\
		\vdots \\
		\rho_{i,\bar{s}-1} \\
		\pi_{i,{e}_1} \\
		\vdots \\
		\pi_{i,{e}_{\bar{n}-\bar{d}}}
		\end{bmatrix} =
		\begin{bmatrix}
		\zz{\sigma_{i,0}(\cJ_a)} \\
		\zz{\sigma_{i,1}(\cJ_a)} \\
		\vdots  \\
		\sigma_{i,\bar{r}-1}(\zz{\cJ_a})
		\end{bmatrix}.
		\label{eq:oa-msr-mx-iua-rack}
	\end{align} 
	Therefore, for any $\zz{\cJ_a}\subseteq{\cJ}$ and every $i\in\cI(\zz{\cJ_a})$, the values $\{\rho_{i,p}\mid p = 1,\ldots,\bar{s}-1 \}$ and $\{\sum_{{g}=0}^{u-1}c_{{e}u+{g},i} \mid {e}\in{\cJ} \}$ can be found
 from the values $\{\sigma_{i,w}(\zz{\cJ_a})\mid w = 0, \ldots, \bar{r}-1 \}$. 
 It follows that we can recover the values $\{\rho_{i,p}\mid p = 1,\ldots,\bar{s}-1 \}$ and $\{\sum_{{g}=0}^{u-1}c_{{e}u+{g},i} \mid {e}\in{\cJ} \}$ for all $i\in\cI_a$.

	Note that for $i\in\cI(\zz{\cJ_a})$, ${e} \in \zz{\cJ_a}\setminus{\cJ}_1$, and for $p \neq 0$, we have $i({e},p) \in \cI_{a-1}$. 
	By the induction hypothesis, we have recovered the values $\{\sum_{{g}=0}^{u-1}c_{{e}u+{g},i} \mid i\in\cI_{a-1};{e}\in{\cJ}\setminus\zz{\cJ_1} \}$, and therefore, we know the values $\{\sum_{{g}=0}^{u-1}c_{{e}u+{g},i({e},p)} \mid j\in\zz{\cJ_a}\setminus\zz{\cJ_1},p\neq 0\}$ for each $i\in\cI_{a}$. 
	With these values and $\{\rho(i,p) \mid i\in\zz{\cI_{a}},p=1,\ldots,s-1 \}$, we can obtain the values $\{\sum_{{g}=0}^{u-1}c_{{e}_1 u+{g},i} \mid p=1,\ldots,s-1 \}$. Since the values of local nodes $\{c_{{e}_1 u+{g},i} \mid {g}\neq {g}_1 \}$ are available, we can further recover the value $c_{j_1,i}$.
	
	Thus, we can obtain the values $\{c_{j_1,i({e}_1,p)}\mid p=0,\ldots,\bar{s}-1 \}$ and 
	$\{\sum_{{g}=0}^{u-1}c_{{e}_1 u+{g},i} \mid {e}\in{\cJ}\setminus{\cJ}_1 \}$ for every $i\in\cI_a$. 
	It follows that we can recover these values 
	for every $i\in\zz{\cI_{a}}$ and $1\le a \le \bar{n}-\bar{d}$ from the helper racks $\zz{\cR}$. In conclusion, we can recover the values $\{c_{j_1,i({e}_1,p)}\mid i\in\cI,p=0,\ldots,\bar{s}-1 \}=\{c_{j_1,i}\mid i=0,\ldots,l-1\}$ from the information obtained from the
	helper racks in $\cR.$
	
	Now let us count the number of symbols we access in each helper rack. It is clear from the definition of $\sigma_{i,w}(\zz{\cJ_a})$ 
that we need to access the symbols $\{c_{{e} u+{g},i}\mid 0\le {g}<u,i\in\cI\}$ for each \zz{${e}\in{\cR}$}. In other words,
we need to access $\bar{s}^{\bar{n}-1}=l/\bar s$ symbols on each node in the helper racks; thus, the total number of accessed symbols equals
$\bar d u l/s.$ Moreover, the set of symbols we access in each helper rack depends on index of the \zz{host rack} but not the index of the helper rack.

Note also that the symbols downloaded to the rack $e_1$ from any helper rack $e\in\cR$ form the subset $\{\sum_{{g}=0}^{u-1}c_{{e} u+{g},i} \mid i\in\cI\}.$ Thus, the total amount of information downloaded for the purposes of repair equals
  $$
  \bar{d}|\cI|=\bar{d}\bar{s}^{\bar{n}-1}=\frac{\bar{d}l}{\bar{d}-\bar{k}+1}.
  $$
This is the smallest possible number according to the bound \eqref{eq:rack-bound}, and thus the codes support optimal repair.

\vspace*{.2in}	\emph{MDS property.}

We will show that the contents of any $n-r$ nodes suffices to find the values of the remaining $r$ nodes.
	
	Let $\zz{\cK}=\{j_1,\ldots,j_r\}\subseteq\{0,\ldots,n-1\}$ be the set of $r$ nodes to be recovered from the set of $n-r$ nodes in $[0,n-1]\setminus\zz{\cK}$. 
	Let us write $j_b={e}_b u+{g}_b$ where $0\le {g}_b< u-1$ for $b=1,\ldots,r$.
	
	Let $\zz{\cJ}$ be the set of distinct ${e}_b,b=1\ldots,r$. For $1 \le a \le |{\cJ}|$, let $\zz{\cJ_a}\subseteq{\cJ}$ be such that $|\zz{\cJ_a}| = a$.  


	Let $\zz{\cI_0}=\{ i=(i_{\bar{n}-1},\ldots,i_0)\in\zz{\{0,\ldots,l-1\}} \mid i_{{e}} \neq 0, {e}\in{\cJ} \}$. 	
	For $1 \le a \le |{\cJ}|$ and $\zz{\cJ_a}\subseteq{\cJ}$, let $\cI(\zz{\cJ_a})=\{ i=(i_{\bar{n}-1},\ldots,i_0)\in\zz{\{0,\ldots,l-1\}} \mid i_{e}=0,{e}\in\zz{\cJ_a}; i_{{e}'} \neq 0, {e}'\in{\cJ}\setminus\zz{\cJ_a} \}$. Let $\cI(a)=\bigcup_{\zz{\cJ_a}\subseteq{\cJ}}\cI(\zz{\cJ_a})$ where $1\le a \le |{\cJ}|$. Observe that \zz{the sets $I_a,0\leq a\leq |\cJ|$ partition the set $\{0,1,\dots,l-1\}$}.
	
	We will prove by induction that we can recover the nodes in $\cJ$ from the nodes in $\{0,1,\dots,n-1\}\setminus\cJ$.
	First, let us establish the induction basis, i.e., we can recover the values $\{c_{j,i}\mid j\in\cJ \}$ for every $i\in\zz{\cI_0}$ from the nodes $\{c_j \mid j\in\zz{\cJ^{c}} \}$. From \eqref{eq:oa-msr-pc-rack}, for $i\in\zz{\cI_0}$, we have
	\begin{align}
		\sum_{j\in\cJ}\lambda_{j}^t c_{j,i}
		&= - 
		\sum_{\zz{j\in\cJ^{c}}}
		\zz{\Big(
		\lambda_{j}^t c_{j,i}
		+ 
		\delta(i_e)
		\sum_{p=1}^{\bar{s}-1}\mu_{p}^t c_{j,i({e},p)}
		\Big)}.
		\label{eq:mds-rack}
	\end{align}
	
	\zz{To simplify notation, denote the right-hand side of \eqref{eq:mds-rack} by $ \sigma_{i,t}=\sigma_{i,t}(\emptyset).$} 
	Note that the value of $\zz{\sigma_{i,t}}$ only depends on the nodes $\{c_j\mid \zz{j\in\cJ^{c}}\}$. \zz{Writing \eqref{eq:mds-rack} in matrix form}, we have
	\begin{align}
		\begin{bmatrix}
		1 & \cdots  & 1 \\
		\lambda_{j_1} & \cdots  & \lambda_{j_{r}} \\
		\vdots & \ddots & \vdots \\
		\lambda_{j_1}^{r-1} & \cdots  & \lambda_{j_{r}}^{r-1} 
		\end{bmatrix}
		\begin{bmatrix}
		c_{j_1,i} \\
		c_{j_2,i} \\
		\vdots \\
		c_{j_r,i}
		\end{bmatrix} =
		\begin{bmatrix}
		\zz{\sigma_{i,0}} \\
		\zz{\sigma_{i,1}} \\
		\vdots  \\
		\zz{\sigma_{i,r-1}}
		\end{bmatrix}.
	\end{align}
	Therefore, the values $\{c_{j,i} \mid j\in\cJ \}$ can be calculated from the values $\{ \zz{\sigma_{i,t}}\mid t = 0, \ldots, r-1\}$ for every $i\in\zz{\cI_0}$.
	
	Now let us establish the induction step. Suppose we recover the values $\{c_{j,i}\mid j\in\cJ \}$ for every $i\in\zz{\cI_{a'}}$ and $0\le a' \le a-1$ from the nodes $\{c_j\mid \zz{j\in\cJ^{c}} \}$, where $1\le a\le |{\cJ}|$.
	
	\zz{Now let us fix a set $\cJ_{a}\subseteq\cJ$ and let} $i\in \cI(\zz{\cJ_a})$. 
	From \eqref{eq:oa-msr-pc-rack}, we have
	\begin{align}
		\sum_{j\in\cJ}\lambda_{j}^t c_{j,i}
		& =
		-
		\sum_{p=1}^{\bar{s}-1}\mu_p^t 
		\sum_{j\in\cJ \colon {e}\in\zz{\cJ_a}}
		c_{j,i({e},p)} 
		- 
		\sum_{\zz{j\in\cJ^{c}}}
		\zz{\Big(
		\lambda_{j}^t c_{j,i}
		+
		\delta(i_e)
		\sum_{p=1}^{\bar{s}-1}\mu_{p}^t c_{j,i({e},p)}
		\Big)}\nonumber\\
	&\zz{=:-\rho'_{i,t}- \sigma_{i,t}(\zz{\cJ_{a}}),} 
		\label{eq:mds-iua-rack}
	\end{align}
	\zz{where the second line serves to introduce the shorthand notation.}
	Note that we know the values $\{\zz{\sigma_{i,t}(\zz{\cJ_a})}\mid t=0,\ldots,r-1 \}$ since the value $\zz{\sigma_{i,t}(\zz{\cJ_a})}$ only depends on the nodes $\{c_j\mid \zz{j\in\cJ^{c}}\}$.
	Furthermore, we also know the values $\{\zz{\rho'_{i,t}}\mid t=0,\ldots,r-1 \}$. Indeed, for $i\in\cI(\zz{\cJ_a})$, ${e} \in \zz{\cJ_a}$, and $p \neq 0$, we have $i({e},p) \in \zz{\cI_{a-1}}$. By the induction hypothesis, we have recovered the values $\{c_{j,i} \mid i\in\zz{\cI_{a-1}}, j\in\cJ \}$, and therefore, we know the values $\{c_{j,i({e},p)} \mid j\in\cJ\colon {e}\in\zz{\cJ_a},p\neq 0 \}$ for each $i\in\zz{\cI_{a}}$. It follows that we know the values $\{\zz{\rho'_{i,t}}\mid t=0,\ldots,r-1 \}$. 
	\zz{Writing \eqref{eq:mds-iua-rack} in matrix form}, we have
	\begin{align}
		\begin{bmatrix}
		1 & \cdots  & 1 \\
		\lambda_{j_1} & \cdots  & \lambda_{j_{r}} \\
		\vdots & \ddots & \vdots \\
		\lambda_{j_1}^{r-1} & \cdots  & \lambda_{j_{r}}^{r-1} 
		\end{bmatrix}
		\begin{bmatrix}
		c_{j_1,i} \\
		c_{j_2,i} \\
		\vdots \\
		c_{j_r,i}
		\end{bmatrix} =
		\begin{bmatrix}
		\zz{\rho'_{i,0}}+\zz{\sigma_{i,0}(\cJ_a)} \\
		\zz{\rho'_{i,1}}+\zz{\sigma_{i,1}(\cJ_a)} \\
		\vdots  \\
		\zz{\rho'_{i,r-1}}+\zz{\sigma_{i,r-1}(\cJ_a)}
		\end{bmatrix}.
	\end{align}
	Therefore, the values $\{c_{j,i} \mid j\in\cJ \}$ can be recovered for every $i\in\cI(\zz{\cJ_a})$ and $\zz{\cJ_a}\subseteq{\cJ}$. It follows that we can recover the values $\{c_{j,i}\mid j\in\cJ \}$ for every $i\in\zz{\cI_{a}}$. Thus, all the values $\{c_{j,i} \mid j\in\cJ,i\in\zz{\cI_{a}},0\le a\le |{\cJ}|\}=\{c_{j,i}\mid j\in\cJ,i\in\{0,\ldots,l-1\}$ can be recovered from the nodes $\{c_j\mid \zz{j\in\cJ^{c}}\}$.
	
	Since $\cJ$ is arbitrary, we conclude that any $n-r$ nodes can recover the entire codeword, i.e., the code is MDS.
	\end{IEEEproof}
	
	\section{A construction of Reed-Solomon codes with optimal repair}\label{sec:RS}
In this section we present a family of scalar MDS codes that support optimal repair of a single node from an arbitrary subset of $\bar d$
helper racks. We still use the same notation as in the previous parts of the paper. As noted earlier, the construction is a modification
of the RS code family in \cite{Tamo18}. The new element of the construction is the idea of coupling the code family of \cite{Tamo18} and
the multiplicative structure that matches the grouping of the nodes into racks. This latter part is similar to the idea of Sec.~\ref{sec:RackCodes}.

Let $q$ be a power of a prime, let $u$ be the size of the rack, and suppose that $u|(q-1).$ Let $k=\bar ku +v, \zz{0}\le v\le u-1,$ $\bar s=\bar d-\bar k+1.$ Let $p_i,i=1,\dots,\bar n$ be distinct primes
such that $p_i\equiv 1\,\text{mod}\, \bar s$ and $p_i>u$ for $i=1,\dots,\bar n$; for instance, we can take the {\em smallest} $\bar n$
primes with these properties. For $i=1,\dots,\bar n$ let $\lambda_i$ be an element of degree $p_i$ over ${\mathbb F}_q$. Let 
  \begin{gather*}
  {F_i}:={\mathbb F_q}(\lambda_j, j\in\{1,\dots,\bar n\}\backslash\{i\}), \;i=1,\dots,\bar n\\
  {\mathbb F}:={\mathbb F_q}(\lambda_1,\dots,\lambda_{\bar n}).
  \end{gather*}
Let ${\mathbb K}$ be an extension of $\mathbb F$ of degree $\bar s$ and let $\mu\in\mathbb{K}$ be a generating element of $\mathbb K$ over $\mathbb{F}$. Thus, for any $i=1,\dots,\bar n$ we have the chain of inclusions
  $$
  {\mathbb F}_q\subset F_i\subset {\mathbb K};
  $$
so $\mathbb{K}$ is the $l$-th degree extension of $\mathbb{F}_q$, where $l=[\mathbb{K} : \mathbb{F}_q] = \bar{s}\prod_{m=1}^{\bar{n}}p_m$. 

Further, let $\lambda\in{\mathbb F}_q$ be an element of multiplicative order $u.$ Consider the set of elements
  $$
  \lambda_{ij}=\lambda_i \lambda^{j-1}, i=1,\dots,\bar n; j=1,\dots,u.
 $$ 
Consider an RS code $\cC=RS_{\mathbb K}(n,k,\Omega)$ where the set of evaluation points $\Omega$ is as follows:
  $$
  \Omega=\bigcup_{i=1}^{\bar n} \Omega_i, \text{ where } \Omega_i=\{\lambda_{ij},j=1,\dots,u\}.
  $$
A codeword of $\cC$ has the form $c=(c_1,c_2,\dots,c_n),$ where the coordinate $c_{m},m=(i-1)u+j,1\le i\le \bar n; 1\le j\le u$ corresponds to the evaluation point $\lambda_{ij}.$

To describe the repair procedure, we will need the following easy modification of Lemma 1 of \cite{Tamo18}.
	\begin{lemma}
		\label{le:rs-subspaces}
		For $i\in \{1,\ldots,\bar{n}\}$, there exists subspace $S_{i}$ of $\mathbb{K}$ such that
		\begin{align}
		\dim_{F_{i}}S_{i}=p_i,\quad  
		 S_{i}+S_i\lambda_i^{u}+\dots+S_i\lambda_{i}^{u(\bar s-1)}
		=\mathbb{K}\label{eq:S}
		\end{align}
where $S_i\beta=\{\gamma \beta, \gamma\in S_i\}$ and the operation $+$ is the Minkowski sum of sets, $T_1+T_2:=\{\gamma_1+\gamma_2: \gamma_1\in T_1,\gamma_2\in T_2\}.$
	\end{lemma}
	\begin{proof}
The space $S_i$ is constructed as follows. Define the following vector spaces over $F_i$:
  \begin{gather*}
  S_i^{(1)}=\text{Span}_{F_i}(\mu^t\lambda_i^{t+e\bar s}, t=0,1,\dots,\bar s-1; e=0,1,\dots,\textstyle{\zz{\frac{p_i-1}{\bar s}-1}})\\
  S_i^{(2)}=\text{Span}_{F_i}\Big(\sum_{t=0}^{\bar s-1}\mu^t\lambda_i^{p_i-1}\Big)  
  \end{gather*}
  and take
  $$
   S_i=S_i^{(1)}+S_i^{(2)}.
   $$
Now the proof of \cite[Lemma 1]{Tamo18} can be followed step by step, using the fact that $\{1,\lambda_{\zz{i}}^{u},\ldots,(\lambda_{\zz{i}}^u)^{\zz{p_i}-1}\}$ 
forms a basis for $\mathbb{F}$ over $\zz{F_i}$, and we do not repeat it here. 
	\end{proof}

%
%
%
%
%
%
%
%
%
%
	
	
The main result of this section is given in the following proposition.
\begin{proposition}
The code $\cC$ supports optimal repair of a single failed node in any rack from any $\bar d$ helper racks.
\end{proposition}
The proof follows the scheme in \cite{Tamo18} which is itself an implementation of the framework for repair of RS codes proposed in \cite{Guruswami16}.
\begin{proof} Let 
  $$
  c=((c_{(i-1)u+j})_{1\le i\le\bar n; 1\le j\le u})
  $$
  be a codeword of $\cC.$ Suppose that
$c_{(i^\ast -1)u+j^\ast}$ is the failed node, i.e., the index of the host rack is $i^\ast, 1\le i^\ast\le n,$ 
and the index of the failed node in this rack is $j^\ast, 1\le j^\ast\le u.$ Denote by 
$\cR \subseteq\{1,\ldots,\bar{n}\}\setminus\{i^\ast\},|\cR|=\bar{d}$ the set of helper racks. The repair relies on the information downloaded
from all the nodes in $\cR$ and the functional nodes in the host rack. Define the annihilator polynomial of the set of
locators of all the nodes in $\cR$:
	 \begin{align}\label{eq:h}
	 	h(x) = \prod_{\substack{i\in\{1,\ldots,\bar{n}\}\setminus(\mathcal{R}\cup\{i^\ast\}),\\1\leq j\leq u}} (x-\lambda_{ij}).
	 \end{align}
	 Let $t=uw$, where $w = 0,\ldots,\bar{s}-1$. Since 
	 \begin{equation}\label{eq:degree}
	 \deg x^th(x)\leq (\bar{s}-1)u+(\bar{n}-\bar{d}-1)u = (\bar{r}-1)u < \bar{r}u-v=n-k
	 \end{equation}
evaluations of the polynomials $x^t h(x)$ are contained in the dual code $\cC^\perp.$ 

Since $\cC^\perp$ itself is a
(generalized) RS code, there is a vector $a=(a_1,\dots,a_n)\in({\mathbb K}^\ast)^n$ such that any codeword of $\cC^\perp$ has the form
$(a_{ij}f(\lambda_{ij}))_{1\le i\le\bar n; 1\le j\le u},$ where $f\in {\mathbb K}[x]$ is a polynomial of degree $\le r-1.$
Thus, by \eqref{eq:degree}, the vector $( a_1\lambda_{11}^t h(\lambda_{11}),\ldots,a_n\lambda_{\bar{n},u}^t h(\lambda_{\bar{n},u}) )\in\cC^{\perp},$
so the inner product of this vector and the codeword $c$ is zero. In other words, we have
	 \begin{align*}
	 	\sum_{j=1}^u a_{(i^\ast-1)u+j}\lambda_{i^\ast j}^t h(\lambda_{i^\ast j}) c_{(i^\ast-1)u+j} 
	 	& = -\sum_{\substack{i=1,\\i\neq i^\ast}}^{\bar{n}}\sum_{j=1}^{u} a_{(i-1)u+j}\lambda_{ij}^t h(\lambda_{ij}) c_{(i-1)u+j}.
	\end{align*}
	Let $S_{i^\ast}$ be the subspace defined in Lemma~\ref{le:rs-subspaces} and let $\{e_1,\ldots,e_{p_{i^\ast}}\}$ be a basis of 
$S_{i^\ast}$ over $F_{i^\ast}.$ Then for $m=1,\ldots,p_{i^\ast}$, we have
	\begin{align*}
		\sum_{j=1}^u 
		e_{m} a_{(i^\ast-1)u+j}\lambda_{i^\ast j}^t h(\lambda_{i^\ast j}) c_{(i^\ast-1)u+j} 
		& = -\sum_{\substack{i=1,\\i\neq i^\ast}}^{\bar{n}}\sum_{j=1}^{u} 
		e_{m} a_{(i-1)u+j}\lambda_{ij}^t h(\lambda_{ij}) c_{(i-1)u+j}.
	\end{align*}	
Evaluating the trace from $\mathbb{K}$ to $F_{i^\ast}$ on both sides of the last equation, we obtain
	\begin{align}
		\tr{i^\ast}
		\Big(\sum_{j=1}^u 
		e_{m} a_{(i^\ast-1)u+j}\lambda_{i^\ast j}^t h(\lambda_{i^\ast j}) &c_{(i^\ast-1)u+j}
		\Big)\nonumber\\
		& = 
		-\sum_{\substack{i=1,\\i\neq i^\ast}}^{\bar{n}}
		\sum_{j=1}^{u}
		\lambda_{ij}^t  h(\lambda_{ij})
		\tr{i^\ast}(e_{m} a_{(i-1)u+j} c_{(i-1)u+j})\nonumber\\
		& = 
		-\sum_{i\in\mathcal{R}}
		\sum_{j=1}^{u}
		\lambda_{ij}^t h(\lambda_{ij})
		\tr{i^\ast}(e_{m} a_{(i-1)u+j}  c_{(i-1)u+j}), \label{eq:tr}
	\end{align}
where we used \eqref{eq:h}, the fact that $\lambda_{ij}\in F_{i^\ast}$ for all $i\ne i^\ast$, and where $t=uw,w=0,\ldots,\bar{s}-1$ and $m=1,\ldots,p_{i^\ast}$.

Recall that $\lambda_{ij}=\lambda_{i}\lambda^{j-1}$ and $\lambda^u=1$. From \eqref{eq:tr} we have
	\begin{align}
		\tr{i^\ast}	\Big(e_{m}\lambda_{i^\ast}^{uw}\sum_{j=1}^u a_{(i^\ast-1)u+j} h(\lambda_{i^\ast j}) c_{(i^\ast-1)u+j}\Big)
		  =	-\sum_{i\in\mathcal{R}}	\lambda_{i}^{uw} \sum_{j=1}^{u} h(\lambda_{ij})	\tr{i^\ast}(e_{m} a_{(i-1)u+j}  c_{(i-1)u+j}),
		  \label{eq:tr-repair}
	\end{align}
	where the parameters $t,m$ are as above.
By \eqref{eq:S} in Lemma~\ref{le:rs-subspaces} and the definition of the set $\{e_m\}$, the set $\{e_{m}\lambda_{i^\ast}^{uw} \mid 1\leq m \leq p_{i^\ast},0\leq w\leq \bar{s}-1  \}$ forms a basis for $
\mathbb{K}$ over $F_{i^\ast}$. 
Therefore, the mapping
   \begin{equation}\label{eq:trace}
   \beta\mapsto \tr{i^\ast}(e_{m}\lambda_{i^\ast}^{uw}\beta), 1\leq m \leq p_{i^\ast},0\leq w\leq \bar{s}-1
   \end{equation}
is a bijection.

The repair procedure is accomplished as follows. For every $i\in \cR,$ the elements 
  \begin{equation}\label{eq:download}
  \sum_{j=1}^{u} h(\lambda_{ij})	\tr{i^\ast}(e_{m} a_{(i-1)u+j}  c_{(i-1)u+j}), m=1,\dots,p_{i^\ast}
  \end{equation}
are downloaded from helper rack $i$. By \eqref{eq:tr-repair} this enables us to find the elements
  $$
  \tr{i^\ast}	\Big(e_{m}\lambda_{i^\ast}^{uw}\sum_{j=1}^u a_{(i^\ast-1)u+j} h(\lambda_{i^\ast j}) c_{(i^\ast-1)u+j}\Big), \quad m=1,\dots,p_{i^\ast}.
  $$
   Next, on account of \eqref{eq:trace} we can find the value of $\sum_{j=1}^u a_{(i^\ast-1)u+j} h(\lambda_{i^\ast j}) 
c_{(i^\ast-1)u+j}.$ Finally, since the values of the coordinates $c_{(i^\ast-1)u+j},j\ne j^\ast$ stored on the functional nodes in the host rack $i^\ast$ are available and the entires of the vector $a$ are nonzero, we
can find $c_{(i^\ast-1)u+j^\ast},$ completing the repair. 

%
	
	 The number of field symbols of $\zz{F_{i^\ast}}$  \eqref{eq:download} transmitted from the helper racks to the host rack equals $\bar{d}\zz{p_{i^\ast}}$. Therefore, we conclude that the repair bandwidth of $\mathcal{C}$ is 
	\begin{align}
		\frac{\bar{d}\zz{p_{i^\ast}}}{[\mathbb{K}: \zz{F_{i^\ast}}]}l & = \frac{\bar{d}l}{\bar{s}}.
	\end{align}
This meets the bound \eqref{eq:rack-bound} with equality, and proves the claim of optimal repair.
\end{proof}
	
\appendices
\addtocontents{toc}{\protect\setcounter{tocdepth}{0}}
\section{Proof of Proposition~\ref{prop:uniform}}\label{sec:uniform}
Let $\cI\subset \cR, |\cI|=\bar k-1$ be a subset of helper racks. Since $\zz{(\bar{d}-\bar k+1)u\ge d-k+1},$ Lemma \ref{lemma:MDS} implies that
  \begin{equation}\label{eq:l}
  \sum_{i\in \cR\backslash\cI}\beta_{i}\ge l.
  \end{equation}
Let us sum the left-hand side on all $\cI\subset\cR, |\cI|=\bar k-1:$
 $$
 \sum_{\begin{substack}{\cI\subset \cR\\|\cI|=\bar k-1}\end{substack}} \sum_{i\in \cR\backslash\cI}\beta_{i}=
 \sum_{i\in\cR}\sum_{\begin{substack}{\cI\subset \cR\\ \cI\not\ni i}\end{substack}}\beta_{i}=\binom{\bar d-1}{\bar k -1}\sum_{i\in \cR}\beta_{i}.
 $$
Together with \eqref{eq:l} we obtain 
$$
\binom{\bar d-1}{\bar k -1}\sum_{i_t\in \cR}\beta_{i}\ge \binom{\bar d}{\bar k-1}
$$
or
$$
\sum_{i\in \cR}\beta_{i}\ge \frac{\bar d l}{\bar d-\bar k+1},
$$
i.e., \eqref{eq:rack-bound}. Moveover, this bound holds with equality if and only if \eqref{eq:l} holds with equality for every $\cI\subset \cR,
|\cI|=\bar k-1.$ Suppose for the sake of contradiction that the uniform download claim does not hold, and there is a rack $i$ such that 
$\beta_{i}\ne l/\bar s,$ for instance, $\beta_i< l/\bar s,$ where $\bar s=\bar d-\bar k+1$. Let $\zz{\cJ\subset \cR},|\cJ|=\bar s,i\in \cJ.$ There must be a rack $i_1\in \cJ$ that 
contributes more than the average number of symbols, i.e., $\beta_{i_1}>l/\bar s.$ Consider the subset $(\cJ\backslash\{i\})\cup \{i_2\},$
where $i_2\ne i$ is another element of $\cR$ (which exists since \zz{$\bar{k}>1$ implies $|\cJ|<|\cR|$}). We have that $\beta_{i_2}<l/\bar s.$
Now take the subset $I=(\cJ\backslash\{i_1\})\cup\{i_2\}$ and note that for it, \eqref{eq:l} fails to hold with equality, a contradiction.

\section{Proof of Prop.~\ref{prop:optimal-access}}\label{app:optimal-access}
\begin{IEEEproof} Let $m'\in \cR$ and let $\cI$ be a subset of $u-v$ nodes in rack $m',$ where $\zz{0}\le v\le u-1.$ Let $\cJ\subseteq \cR\backslash\{m'\},|\cJ|=\bar d-\bar k$ be a subset of helper racks. These racks contain $u(\bar d-\bar k)=d-k+1-(u-v)$ nodes in these racks, and together with the nodes in the set
$\cI$ this forms a group of $d-k+1$ nodes. Using Lemma \ref{lemma:MDS}, we have
  \begin{equation}\label{eq:MDS-access}
  \sum_{m\in \cJ}\sum_{e=1}^u \alpha_{m,e}+\sum_{e\in \cI}\alpha_{m',e}\ge l.
  \end{equation}
Let us average over the $\binom{u}{u-v}$ choices of the set $\cI:$ 
   $$
   \binom u{u-v}\sum_{m\in\cJ}\sum_{e=1}^u \alpha_{m,e}+ \sum_{\cI:|\cI|=u-v}\sum_{e\in\cI}\alpha_{m',e}\ge \binom u{u-v}l.
   $$
Interchanging the sums in the second term on the left, we obtain
  $$
  \binom u{u-v} \sum_{m\in\cJ}\sum_{e=1}^u \alpha_{m,e} +\binom{u-1}{u-v-1}\sum_{e=1}^u\alpha_{m',e}\ge \binom u{u-v}l.
  $$
or
  \begin{equation}\label{eq:u-v}
  \frac u{u-v}\sum_{m\in\cJ}\sum_{e=1}^u\alpha_{m,q}+\sum_{e=1}^u\alpha_{m',e}\ge\frac{u}{u-v}l.
  \end{equation}
Now let us average over the choice of $\cJ\subset \cR\backslash\{m'\}.$ Noting that
  $$
  \sum_{\substack{\cJ\subset\cR\setminus\{m'\}\\|\cJ|=\bar{s}-1}}\sum_{m\in\cJ}\sum_{e=1}^{u}\alpha_{m,e}
  =\zz{\binom{\bar{d}-2}{\bar{s}-2}}
		 	\sum_{m\in\cR\setminus\{m'\}}
		 	\sum_{e=1}^{u}\alpha_{m,e}
			$$
we obtain from \eqref{eq:u-v}	
  $$
  \frac{u}{u-v}\binom{\bar{d}-2}{\bar{s}-2}
		 	\sum_{m\in\cR\setminus\{m'\}}
		 	\sum_{e=1}^{u}\alpha_{m,e}+\binom{\bar d-1}{\bar s-1}\sum_{e=1}^u\alpha_{m',e}\ge\binom{\bar d-1}{\bar s-1}\frac{u}{u-v}l.
  $$
 On account of the assumption that $\bar d\ge 2,\bar s\ge 2$ we find
  \begin{align}\label{eq:avg-2}
		 	\frac{u}{u-v}
		 	\sum_{m\in\mathcal{R}\setminus\{m'\}}
		 	\sum_{e=1}^{u}\alpha_{m,e}+
		 	\frac{\bar{d}-1}{\bar{s}-1}\sum_{e=1}^{u}\alpha_{m',e}
		 	& \geq \frac{\bar{d}-1}{\bar{s}-1}\frac{u}{u-v}l.
		 \end{align} 
Now let us average on the choice of $m'\in\cR:$		
  $$
  \frac{u(\bar{d}-1)}{u-v}
			 \sum_{m\in\mathcal{R}}
			 \sum_{e=1}^{u}\alpha_{m,e}+
			 \frac{\bar{d}-1}{\bar{s}-1}\sum_{m'\in\mathcal{R}}\sum_{e=1}^{u}\alpha_{m',e}
			  \geq \frac{\bar{d}-1}{\bar{s}-1}\frac{u}{u-v}\bar{d}l
  $$
  or 
  $$
  \alpha:=\sum_{m\in \cR}\sum_{e=1}^u \alpha_{m,e}\ge \frac{\bar{d}ul}{u(\bar{s}-1)+u-v}
		 	= \frac{\bar{d}ul}{s}.
  $$
  Equality holds if and only if it holds in \eqref{eq:MDS-access}. This implies the uniform access condition, which is proved in exactly
  the same way as the uniform download condition in Prop.~\ref{prop:uniform}.
 \end{IEEEproof}

\section{Proof of Theorem~\ref{thm:bound} }
	{\em Proof of Part (a):} The proof will be given for the repair of a node in a systematic rack. In the end we will argue that the claimed bound also applies to the repair of nodes in parity racks.

Without loss of generality, assume $\{\bar{k}+1,\ldots,\bar{k}+\bar{s}\}$ to be the $\bar{s}$ parity racks that are involved in the 
repair of the failed node. Let $\bar k+i, {i}=1,\ldots,\bar{s}$ be a helper rack. Since the repair scheme is linear, the information
that this rack provides is obtained through a linear transformation of its contents. Denote the matrix of this transformation by 
 $S_{\bar{k}+{i},m_1}$ and call it the {\em repair matrix} for repairing a failed node in rack $m_1$ from rack $\bar{k}+{i}$ (and call its row space the {\em repair subspace} of the node.
 Note that it is an $\frac{l}{\bar{s}}\times ul$ matrix over $F$; moreover, {for optimal repair, the rank of $S_{\bar{k}+{i},m_1}$ necessarily is $l/\bar{s}$ for all ${i}\in[\bar{s}]$. 
The information that parity rack $\bar{k}+{i}$ transmits to repair the failed node in rack $m_1$ is given by 
	\begin{align}\label{eq:repair}
		S_{\bar{k}+{i},m_1}\bfc_{\bar{k}+{i}} & = S_{\bar{k}+{i},m_1} \Big(c_{k+({i}-1)u+1},\ldots,c_{k+{i}u} \Big)^T\nonumber\\
		& = S_{\bar{k}+{i},m_1} \Big(\sum_{j=1}^{k}A_{({i}-1)u+1,j}c_j,\ldots,\sum_{j=1}^{k}A_{{i}u,j}c_j \Big)^T, i=1,\dots,\bar s.
	\end{align} 
	
For given $i,j$ let us define the block matrix $\cA_{{i},j}=(A_{({i}-1)u+1,j},\ldots,A_{{i}u,j})^T$ (a part of column $j$ in the encoding
matrix that corresponds to rack $m$). Suppose that 
the index of the failed node in rack $m_1$ is $j_1,$ and note that $(m_1-1)u+1\le j_1\le m_1u$. Then from \eqref{eq:repair} we obtain
	\begin{align}
		S_{\bar{k}+{i},m_1}\bfc_{\bar{k}+{i}}
		& = S_{\bar{k}+{i},m_1} \cA_{{i},j_1}c_{j_1}+S_{\bar{k}+{i},m_1}\sum_{\substack{j=1\\ j\neq j_1}}^{k}\cA_{{i},j_1}c_j,\quad i=1,\dots,\bar s.
		\label{eq:parity-offer-node}
	\end{align}
	From \eqref{eq:parity-offer-node} we observe that the information that parity rack $\bar{k}+{i}$ provides for the repair of node $c_{j_1}$ contains interference from the other systematic nodes $c_{j}, j\neq j_1$. Moreover, as rack $m_1$ collects all the information sent from the helper racks $\bar{k}+{i}, 1\le {i}\le\bar{s}$, in order to repair node $c_j$, it is necessary that 
	\begin{align}
		\rk
		\begin{bmatrix}
		S_{\bar{k}+1,m_1}\cA_{1,j_1}\\
		\vdots\\
		S_{\bar{k}+\bar{s},m_1}\cA_{\bar{s},j_1}
		\end{bmatrix}
		& = l.
		\label{eq:ia-1-pre}
	\end{align}
This relation holds true because Equations \eqref{eq:parity-offer-node} evaluate a linear combination of the contents of the nodes in the host rack. To retrieve the $l$ symbols of the failed node from this linear combination, condition \eqref{eq:ia-1-pre} is necessary.

	Let us further define 
the $ul\times ul$ matrix $\cD_{{i},m}=(\cA_{{i},(m-1)u+1},\ldots,\cA_{{i},mu})$ by assembling together $u$ columns of the form $\cA_{\cdot,\cdot},$ i.e., $\cD_{{i},m}=(A_{\alpha,\beta}), (i-1)u+1\le \alpha\le iu, (m-1)u+1\le \beta\le mu.$ These matrices are defined for notational convenience and enable us to argue about entire racks rather than their elements.
Since the code $\cC$ is MDS, the matrix $\cD_{{i},m}$ is invertible for all $1\le {i}\le \bar{r}$ and $1\le m\le\bar{k}$. 
 Rewriting
\eqref{eq:parity-offer-node} with this notation, we obtain
	\begin{align}
		S_{\bar{k}+{i},m_1}\bfc_{\bar{k}+{i}} & = S_{\bar{k}+{i},m_1}\cD_{{i},m_1}\bfc_{m_1} + S_{\bar{k}+{i},m_1}\sum_{\substack{m=1\\ m\neq m_1}}^{\bar{k}}\cD_{{i},m}\bfc_{m}.
		\label{eq:parity-offer-rack}
	\end{align}
	Since $(m_1-1)u+1\le j_1\le m_1u$, from \eqref{eq:ia-1-pre} we have 
	\begin{align}
		\rk
		\begin{bmatrix}
		S_{\bar{k}+1,m_1}\cD_{1,m_1}\\
		\vdots\\
		S_{\bar{k}+\bar{s},m_1}\cD_{\bar{s},m_1}
		\end{bmatrix}
		& = l.
		\label{eq:ia-1}
	\end{align}
	So far we have only considered the information provided by the parity racks. It remains to characterize the information transmitted by the systematic racks. From \eqref{eq:parity-offer-rack}, in order to cancel out the interference from systematic rack $m\neq m_1$, rack $m_1$ needs to download from systematic rack $m$ the vector $(S_{\bar{k}+1,m_1}\cD_{1,m}\bfc_m,\ldots,S_{\bar{k}+\bar{s},m_1}\cD_{\bar{s},m}\bfc_m)^T$.  
	By Proposition~\ref{prop:uniform}, for optimal repair we necessarily have 
	\begin{align}
		\rk\begin{bmatrix}
		S_{\bar{k}+1,m_1}\cD_{1,m}\\
		\vdots\\
		S_{\bar{k}+\bar{s},m_1}\cD_{\bar{s},m}
		\end{bmatrix} & = \frac{l}{\bar{s}}.
		\label{eq:ia-2}
	\end{align}
The rank conditions \eqref{eq:ia-1} and \eqref{eq:ia-2} 
give rise to the following subspace conditions. For any $m_1\in[\bar{k}]$,
	\begin{align}
		\bigoplus_{{i}=1}^{\bar{s}}\langle S_{\bar{k}+{i},m_1}\cD_{{i},m_1}\rangle & = F^l,\label{eq:ia-sp-1}\\
		\langle S_{\bar{k}+1,m_1}\cD_{1,m}\rangle   =\ldots & = \langle S_{\bar{k}+\bar{s},m_1}\cD_{\bar{s},m}\rangle ,\;\;m\neq m_1.\label{eq:ia-sp-2}
	\end{align}
The proof of the lower bounds in the theorem relies on these necessary conditions. Let us first bound the dimension 
of the intersection of the row spaces of the repair matrices.
	
\begin{lemma}\label{lemma:intersection} Let ${\cJ}\subset [\bar k]$ be a subset of systematic nodes such that $|{\cJ}|\le k-1$ and
$m_1\in {\cJ}$. Then for any $i,i'\ge 1$
  \begin{equation}
  \dim \bigcap_{m\in{\cJ}}\langle S_{\bar{k}+{i},m}\rangle  = \dim \bigcap_{m\in{\cJ}}\langle S_{\bar{k}+{i'},m}\rangle.
  \label{eq:same-dim-sp-parity-rack1}
  \end{equation}
\end{lemma}
\begin{IEEEproof}
Let $a\in[\bar{k}]\setminus{\cJ}$. Since $\cD_{{i},a}$ is nonsingular, for each ${i}\in[\bar{s}]$ we have 
	\begin{align*}
		\bigcap_{m\in{\cJ}} \langle S_{\bar{k}+{i},m}\cD_{{i},a}\rangle & = \Big(\bigcap_{m\in{\cJ}}\langle S_{\bar{k}+{i},m}\rangle\Big)\cD_{{i},a}.
	\end{align*}
On the previous line we take the intersection of the subspaces as indicated, then write a basis of the resulting subspace into rows of a matrix, which has $ul$ columns. Since this is also the number of rows of $\cD_{i,a},$ this operation is well defined.

Since $a\notin{\cJ}$, for any $i,{i'}\in[\bar{s}]$, from \eqref{eq:ia-sp-2} we have 
	\begin{align*}
		\bigcap_{m\in{\cJ}} \langle S_{\bar{k}+{i},m}\cD_{{i},a}\rangle & =
		\bigcap_{m\in{\cJ}} \langle S_{\bar{k}+{i'},m}\cD_{{i'},a}\rangle.
	\end{align*}
	Therefore, for any ${i},{i'}\in[\bar{s}]$
	\begin{align}
		\Big(\bigcap_{m\in{\cJ}}\langle S_{\bar{k}+{i},m}\rangle\Big)\cD_{{i},a} & = \Big(\bigcap_{m\in{\cJ}}\langle S_{\bar{k}+{i'},m}\rangle\Big)\cD_{{i'},a}.
	\end{align}
	Since $\cD_{{i},a}$ and $\cD_{{i'},a}$ are invertible, \eqref{eq:same-dim-sp-parity-rack1} follows.
\end{IEEEproof}		
		
\vspace*{.1in}	Now consider the subspace $\bigcap_{m\in{\cJ}} \langle S_{\bar{k}+{i},m}\cD_{{i},m_1}\rangle,$ where ${\cJ}$ is as in Lemma 
\ref{lemma:intersection}. For any ${i'}\in[\bar{s}]$, we have
	\begin{align}
\bigcap_{m\in{\cJ}} \langle S_{\bar{k}+{i},m}\cD_{{i},m_1}\rangle & = \langle S_{\bar{k}+{i},m_1}\cD_{{i},m_1}\rangle  \bigcap \Big( 
\bigcap_{m\in{\cJ}\setminus\{m_1\}}\langle  S_{\bar{k}+{i},m}\cD_{{i},m_1}\rangle  \Big)\nonumber\\
		& = \langle S_{\bar{k}+{i},m_1}\cD_{{i},m_1}\rangle  \bigcap \Big( \bigcap_{m\in{\cJ}\setminus\{m_1\}} \langle S_{\bar{k}+{i'},m}\cD_{{i'},m_1} \rangle \Big)\nonumber\\
		& \subseteq \bigcap_{m\in{\cJ}\setminus\{m_1\}} \langle S_{\bar{k}+{i'},m}\cD_{{i'},m_1}\rangle .
		\label{eq:sp-chain-parity-rack}
	\end{align}
	Summing on ${i}\in[\bar{s}]$ on both sides of \eqref{eq:sp-chain-parity-rack}, we obtain
	\begin{align}\label{eq:ds}
		\bigoplus_{{i}=1}^{\bar{s}}\Big( \bigcap_{m\in{\cJ}} \langle S_{\bar{k}+{i},m}\cD_{{i},m_1}\rangle \Big) & \subseteq \bigcap_{m\in{\cJ}\setminus\{m_1\}} \langle S_{\bar{k}+{i'},m}\cD_{{i'},m_1}\rangle .
	\end{align}
Note that this is a direct sum because the subspaces $\bigcap_{m\in{\cJ}}\langle S_{\bar{k}+{i},m}\cD_{{i},m_1}\rangle$ form a subset of the set of subspaces on the left-hand side of \eqref{eq:ia-sp-1} and therefore (for different $i$) are disjoint.

	Since $\cD_{{i},m_1}$ and $\cD_{{i'},m_1}$ are invertible, we conclude that
	\begin{align*}
		\sum_{{i}=1}^{\bar{s}}\dim  \bigcap_{m\in{\cJ}} \langle S_{\bar{k}+{i},m}\rangle  & \le \dim  \bigcap_{m\in{\cJ}\setminus\{m_1\}} \langle S_{\bar{k}+{i'},m}\rangle .
	\end{align*}
	Note that from \eqref{eq:same-dim-sp-parity-rack1}, we have
	\begin{align*}
		\sum_{{i}=1}^{\bar{s}}\dim  \bigcap_{m\in{\cJ}} \langle S_{\bar{k}+{i},m}\rangle  & = \bar{s}\dim  \bigcap_{m\in{\cJ}} \langle S_{\bar{k}+{i},m}\rangle .
	\end{align*}
	Therefore,
	\begin{align}
		\dim  \bigcap_{m\in{\cJ}} \langle S_{\bar{k}+{i},m}\rangle & \le \frac{1}{\bar{s}}\dim  \bigcap_{m\in{\cJ}\setminus\{m_1\}}\langle  S_{\bar{k}+{i'},m}\rangle\nonumber\\
		& \le \frac{1}{\bar{s}^{|{\cJ}|-1}}\dim \langle S_{\bar{k}+{i'},m}\rangle\nonumber\\
		& \le \frac{l}{\bar{s}^{|{\cJ}|}}.
		\label{eq:bound-dim-parity}
	\end{align}

If $l\ge \bar{s}^{\bar{k}-1},$ then \eqref{eq:bound1} is proved, so let us assume that $l < \bar{s}^{\bar{k}-1}.$ 
\begin{lemma} Let $\cT\subset [\bar n]\backslash\{\bar k+i\}$ be a subset such that 
\begin{itemize}
\item $1\le |\cT|\le \bar n-1,$
\item $m_1\in\cT,$ 
\item $\cT$ contains $\min(\bar k-1,|\cT|)$ systematic racks.
\end{itemize}  Then 
  \begin{equation}
		\dim  \bigcap_{m\in\cT} \langle S_{\bar{k}+{i},m} \rangle
		 \le \frac{l}{\bar{s}^{|\cT|}}.\label{eq:cT}
	\end{equation}
\end{lemma}
{\em Remark:} Some of the repair matrices in \eqref{eq:cT} refer to the repair scheme of a parity node (a node in a parity rack)
using information from another parity rack. These matrices exist and are well defined because by assumption, the code $\cC$ supports optimal repair
of any node from any set of $\bar d$ helper racks.
\begin{proof}
By the assumption before the lemma, Eq.~\eqref{eq:bound-dim-parity} holds for any ${\cJ}$ of size $\le \bar k-1,$ which proves the claim for the case $|\cT|\le \bar k-1.$ At the same time,
if $|\cT|>\bar k-1,$ take $|{\cJ}|=\bar k-1$ in \eqref{eq:bound-dim-parity} and note that ${\cJ}\subset \cT.$ In this case \eqref{eq:bound-dim-parity} implies 
\eqref{eq:cT} and the proof is complete.
\end{proof}

From \eqref{eq:cT} we observe that the subspaces $\langle S_{\bar{k}+{i},m} \rangle, m\in\cT$ have a vector in common if and only if 
$|\cT|\le \log_{\bar s}l.$ Now consider a $ul\times (\bar n-1)$ matrix $V$ whose rows correspond to the $ul$ vectors in the standard basis of $F^{ul}$ and columns to the repair matrices $S_{\bar k+i,m},m\in[\bar{n}]\setminus\{\bar{k}+{i}\}.$ Put $V_{im}=1$ if the $i$th vector is one of the rows of the $m$th repair matrix and $0$ otherwise. The code has the optimal access property if and only if the rows of the repair matrices are formed of standard basis vectors. Every column of $V$
contains $l/\bar s$ ones, and if $|\cT|\le \log_{\bar s}l,$ then every row contains at most $\log_{\bar s} l$ ones; thus
   $$
\frac{l}{\bar{s}} (\bar{n}-1) \le ul \log_{\bar{s}} l.
  $$
It follows that 
	\begin{align}
		l \geq \bar{s}^{\frac{\bar{n}-1}{s}}, \label{eq:doublecount}
	\end{align}
	where we used $s=\bar{s}u$. This concludes the proof of Part (a).

\vspace*{.05in} {\em Proof of Part (b):}	
We closely follow the arguments in Part (a) with the only difference that the set $\cT$ can now be of size $\bar n,$ which is possible
because the repair matrices are independent of the choice of the helper racks.

Let us outline the argument.  Let $|{\cJ}|=\bar{k}-1$ and $l < \bar{s}^{\bar{k}-1}$.
In this case \eqref{eq:bound-dim-parity} implies that
	\begin{align*}
		\dim \bigcap_{m\in{\cJ}} \langle S_{m}\rangle & =0
	\end{align*}
	(even in the case when the repair scheme is chosen based on the location of the helpers, and all the more so in the current case).
It follows that, for any $\cT\subseteq[\bar{n}]$ such that $\cT\supseteq{\cJ}$, we have
	\begin{align*}
		\dim \bigcap_{m\in\cT} S_{m} & =0.
	\end{align*}
	Therefore, for $l<\bar{s}^{\bar{k}-1}$, we have
	\begin{align}
		\dim  \bigcap_{m\in\cT} S_{m} 
		& \le \frac{l}{\bar{s}^{|\cT|}},
	\end{align}
	for $1\le |\cT|\le \bar{n}$.  For the left-hand side of the above inequality to be greater than $1$, we necessarily have $|\cT| \le \log_{\bar{s}} l$.} 
	
	 Repeating the argument that led to \eqref{eq:doublecount}, we obtain
	\begin{align}
		l \geq \bar{s}^{{\bar{n}}/{s}}.
	\end{align}
	
Thus, we have proved cases (a) and (b) of the theorem for repairing a failed node in a systematic rack. The same bounds hold
for repairing a failed node in any of the parity racks.  Indeed, note that a parity rack $ \bfc_{\bar k+i}, i=1,\dots,\bar r$ 
is computed from the systematic racks as follows:
  \begin{equation}\label{eq:D}
    \bfc_{\bar k+i}=\sum_{j=1}^{\bar k}\cD_{\bar k+i,j}\bfc_{j}
   \end{equation}
If a node in rack $\bar k+i$ has failed, we first choose $\bar d$ helper racks and isolate any $(\bar k-1)$-subset of the chosen $\bar d$-set. Then we write  equations of the form \eqref{eq:D} where on the right we use these $\bar k-1$ racks together with the host rack to express the code symbols in the remaining $\bar r$ racks. These
equations are obtained from \eqref{eq:D} using obvious matrix transformations (no complications arise because the code $\cC$ is MDS). After
that, we can repeat the proofs given above, which establishes our claim.

	\bibliographystyle{IEEEtranS}
	\bibliography{./rfs}

\begin{thebibliography}{10}
\providecommand{\url}[1]{#1}
\csname url@samestyle\endcsname
\providecommand{\newblock}{\relax}
\providecommand{\bibinfo}[2]{#2}
\providecommand{\BIBentrySTDinterwordspacing}{\spaceskip=0pt\relax}
\providecommand{\BIBentryALTinterwordstretchfactor}{4}
\providecommand{\BIBentryALTinterwordspacing}{\spaceskip=\fontdimen2\font plus
\BIBentryALTinterwordstretchfactor\fontdimen3\font minus
  \fontdimen4\font\relax}
\providecommand{\BIBforeignlanguage}[2]{{%
\expandafter\ifx\csname l@#1\endcsname\relax
\typeout{** WARNING: IEEEtranS.bst: No hyphenation pattern has been}%
\typeout{** loaded for the language `#1'. Using the pattern for}%
\typeout{** the default language instead.}%
\else
\language=\csname l@#1\endcsname
\fi
#2}}
\providecommand{\BIBdecl}{\relax}
\BIBdecl

\bibitem{Akhlagi10}
S.~Akhlagi, A.~Kiani, and M.~Ghabavati, ``Cost-bandwidth tradeoff in
  distributed storage systems,'' \emph{Computer Commuunications}, vol.~33,
  no.~17, pp. 2105--2115, 2010.

\bibitem{Balaji18}
S.~Balaji, M.~Krishnan, M.~Vajha, V.~Ramkumar, B.~Sasidharan, and P.~Kumar,
  ``Erasure coding for distributed storage: {A}n overview,'' \emph{Science
  China Information Sciences}, 2018, preprint arXiv:1806.04437, 43pp.

\bibitem{balaji2017tight}
S.~Balaji and P.~V. Kumar, ``A tight lower bound on the sub-packetization level
  of optimal-access {MSR} and {MDS} codes,'' \emph{arXiv preprint
  arXiv:1710.05876}, 2017.

\bibitem{Cadambe13}
V.~R. Cadambe, S.~A. Jafar, H.~Maleki, K.~Ramchandran, and C.~Suh, ``Asymptotic
  interference alignment for optimal repair of {MDS} codes in distributed
  storage,'' \emph{IEEE Trans. Inf. Theory}, vol.~59, no.~5, pp. 2974--2987,
  2013.

\bibitem{dimakis2010network}
A.~G. Dimakis, P.~B. Godfrey, Y.~Wu, M.~J. Wainwright, and K.~Ramchandran,
  ``Network coding for distributed storage systems,'' \emph{IEEE Trans. Inf.
  Theory}, vol.~56, no.~9, pp. 4539--4551, 2010.

\bibitem{Gaston2013}
B.~Gast{\' o}n, J.~Pujol, and M.~Villanueva, ``A realistic distributed storage
  system that minimizes data storage and repair bandwidth,'' in \emph{Proc.
  Data Compression Conf}, 2013, arXiv:1301.1549, 10pp.

\bibitem{Guruswami16}
V.~Guruswami and M.~Wootters, ``Repairing {R}eed-{S}olomon codes,'' \emph{IEEE
  Trans. Inform. Theory}, vol.~63, no.~9, pp. 5684--5698, 2017.

\bibitem{hou2018rack}
H.~Hou, P.~Lee, K.~Shum, and Y.~Hu, ``Rack-aware regenerating codes for data
  centers,'' \emph{arXiv preprint arXiv:1802.04031}, 2018.

\bibitem{hu2016double}
Y.~Hu, P.~P.~C. Lee, and X.~Zhang, ``Double regenerating codes for hierarchical
  data centers,'' in \emph{2016 IEEE International Symposium on Information
  Theory (ISIT)}.\hskip 1em plus 0.5em minus 0.4em\relax IEEE, 2016, pp.
  245--249.

\bibitem{hu2017optimal}
Y.~Hu, X.~Li, M.~Zhang, P.~Lee, X.~Zhang, P.~Zhou, and D.~Feng, ``Optimal
  repair layering for erasure-coded data centers: {F}rom theory to practice,''
  \emph{ACM Transactions on Storage (TOS)}, vol.~13, no.~4, 2017, article \#33.

\bibitem{Kermarrec11}
A.~M. Kermarrec, N.~Le~Scouarnec, and G.~Straub, ``Repairing multiple failures
  with coordinated and adaptive regenerating codes,'' in \emph{2011 Int.
  Sympos. Network Coding (NetCod)}.\hskip 1em plus 0.5em minus 0.4em\relax
  IEEE, 2011, pp. 1--6.

\bibitem{Li14}
J.~Li and B.~Li, ``Cooperative repair with minimum-storage regenerating codes
  for distributed storage,'' in \emph{2014 Proceedings IEEE INFOCOM}.\hskip 1em
  plus 0.5em minus 0.4em\relax IEEE, 2014, pp. 316--324.

\bibitem{LiTangTian18}
J.~Li, X.~Tang, and C.~Tian, ``A generic transformation to enable optimal
  repair in mds codes for distributed storage systems,'' \emph{IEEE Trans. Inf.
  Theory}, vol.~64, no.~9, pp. 6257--6267, 2018.

\bibitem{Pernas2013}
J.~Pernas, C.~Yuen, B.~Gast{\' o}n, and J.~Pujol, ``Non-homogeneous two-rack
  model for distributed storage systems,'' in \emph{Proc. IEEE Int. Symp.
  Information Theory}, 2013, pp. 1237--1241.

\bibitem{PAM18}
N.~Prakash, V.~Abdrashitov, and M.~Medard, ``The storage vs repair bandwidth
  trade-offs for clustered storage systems,'' \emph{IEEE Trans. Inf. Theory},
  vol.~64, no.~8, pp. 5783--5805, 2018.

\bibitem{Rashmi11}
K.~V. Rashmi, N.~B. Shah, and P.~V. Kumar, ``Optimal exact-regenerating codes
  for distributed storage at the {MSR} and {MBR} points via a product-matrix
  construction,'' \emph{IEEE Trans. Inf. Theory}, vol.~57, no.~8, pp.
  5227--5239, 2011.

\bibitem{rawat2016centralized}
A.~S. Rawat, O.~O. Koyluoglu, and S.~Vishwanath, ``Centralized repair of
  multiple node failures with applications to communication efficient secret
  sharing,'' \emph{arXiv preprint arXiv:1603.04822}, 2016.

\bibitem{Sahraei2017a}
S.~Sahraei and M.~Gastpar, ``Increasing availability in distributed storage
  systems via clustering,'' 2017, {P}reprint: arXiv:1710.02653v2.

\bibitem{Sasid16}
B.~Sasidharan, M.~Vajha, and P.~V. Kumar, ``An explicit, coupled-layer
  construction of a high-rate {MSR} code with low sub-packetization level,
  small field size and all-node repair,'' 2016, arXiv:1607.07335.

\bibitem{Shum13}
K.~W. Shum and Y.~Hu, ``Cooperative regenerating codes,'' \emph{IEEE
  Transactions on Information Theory}, vol.~59, no.~11, pp. 7229--7258, 2013.

\bibitem{Sohn18a}
J.-y. Sohn, B.~Choi, and J.~Moon, ``A class of {MSR} codes for clustered
  distributed storage,'' in \emph{Proc. IEEE Intern. Sympos. Inform. Theory,
  June 2018, Vail, CO}, 2018, pp. 2366--2370.

\bibitem{Sohn18}
J.-y. Sohn, B.~Choi, S.~Yoon, and J.~Moon, ``Capacity of clustered distributed
  storage,'' \emph{IEEE Trans. Inf. Theory}, vol.~65, no.~1, pp. 81--107, 2019.

\bibitem{Tamo13}
I.~Tamo, Z.~Wang, and J.~Bruck, ``Zigzag codes: {MDS} array codes with optimal
  rebuilding,'' \emph{IEEE Trans Inf.Theory}, vol.~59, no.~3, pp. 1597--1616,
  2013.

\bibitem{Tamo14}
------, ``Access versus bandwidth in codes for storage,'' \emph{IEEE Trans.
  Inf. Theory}, vol.~60, no.~4, pp. 2028--2037, 2014.

\bibitem{Tamo18}
I.~Tamo, M.~Ye, and A.~Barg, ``The repair problem for {R}eed-{S}olomon codes:
  {O}ptimal repair of single and multiple erasures, asymptotically optimal node
  size,'' \emph{IEEE Trans. Inf. Theory}, 2019, to appear. Preprint available
  at arXiv:1805.01883.

\bibitem{Tebbi2014}
M.~A. Tebbi, T.~H. Chan, and C.~W. Sung, ``A code design framework for
  multi-rack distributed storage,'' in \emph{Proc. IEEE Information Theory
  Workshop (ITW 2014)}, 2014, pp. 55--59.

\bibitem{clay18}
M.~Vajha, V.~Ramkumar, B.~Puranik, G.~Kini, E.~Lobo, B.~Sasidharan, P.~Kumar,
  A.~Barg, M.~Ye, S.~Hussain, S.~Narayanamurthy, and S.~Nandi, ``Clay codes:
  {M}oulding {MDS} codes to yield an {MSR} code,'' in \emph{16th USENIX
  Conference on File and Storage Technologies (FAST 2018), Oakland, CA}, Feb.
  2018, pp. 139--154.

\bibitem{Ye16}
M.~Ye and A.~Barg, ``Explicit constructions of high-rate {MDS} array codes with
  optimal repair bandwidth,'' \emph{IEEE Trans. Inf. Theory}, vol.~63, no.~4,
  pp. 2001--2014, 2017.

\bibitem{Ye19a}
------, ``Cooperative repair: {C}onstructions of optimal {MDS} codes for all
  admissible parameters,'' \emph{IEEE Trans. Inf. Theory}, vol.~65, 2019.

\bibitem{ye2017explicit}
------, ``Explicit constructions of optimal-access {MDS} codes with nearly
  optimal sub-packetization,'' \emph{IEEE Transactions on Information Theory},
  vol.~63, no.~10, pp. 6307--6317, 2017.

\end{thebibliography}
\end{document}